%% file: saw-sa.tex
\def\version{February 23, 2017}

\documentclass[12pt,hidelinks]{article}

\title{
Four-dimensional weakly self-avoiding walk
\\
with
contact self-attraction}

\author{
  Roland Bauerschmidt\thanks{
    Statistical Laboratory,
    DPMMS,
    University of Cambridge,
    Wilberforce Road, Cambridge CB3 0WB, UK.
    {\tt rb812@cam.ac.uk}},
  Gordon Slade\thanks{
    Department of Mathematics,
    University of British Columbia,
    Vancouver, BC, Canada V6T 1Z2.
    {\tt slade@math.ubc.ca},
    {\tt bwallace@math.ubc.ca}},
  and
  Benjamin C. Wallace$^\dagger$}

\date{\version}

\def\macrosPb{}
\def\macrosHarxiv{}
\input{macros}

\renewcommand{\chicCov}{{\vartheta}}
\newcommand{\chicCovgen}{{\tilde\vartheta}}
\renewcommand{\lt}{L}
\newcommand{\nubar}{\bar{\nu}}

\begin{document}

\maketitle

\begin{abstract}
  We consider the critical behaviour of the
  continuous-time weakly self-avoiding
  walk with contact self-attr\-ac\-tion   on $\Z^4$,
  for sufficiently small attraction.
  We prove that the susceptibility and correlation length of
  order $p$ (for any $p>0$) have logarithmic corrections to mean field scaling,
  and that the critical two-point function
  is asymptotic to a multiple of $|x|^{-2}$.
  This shows that small contact self-attraction results in the same
  critical behaviour as no contact self-attraction; a collapse transition is
  predicted for larger self-attraction.
  The proof
  uses a supersymmetric representation of the two-point function,
  and is based on a rigorous renormalisation group method that
  has been used to prove the same
  results for
  the weakly self-avoiding walk, without self-attraction.
\end{abstract}


\section{The model and main result}

The self-avoiding walk is a basic model for a linear polymer chain in a
good solution.
The repulsive
self-avoidance constraint models the excluded volume effect of the polymer.
In a \emph{poor} solution,
the polymer tends to avoid contact with the solution by instead making contact
with itself.  This is modelled by a self-attraction
favouring nearest-neighbour
contacts.
The self-avoiding walk is already a notoriously difficult problem,
and the combination of these two competing tendencies creates additional
difficulties and an interesting
phase diagram.

In this paper, we consider a continuous-time version of the
weakly self-avoiding walk with nearest-neighbour contact self-attraction
on $\Z^4$.
When both the self-avoid\-ance and self-attraction are sufficiently weak,
we prove that the susceptibility and finite-order correlation length
  have logarithmic corrections to mean field scaling
  with exponents $\frac14$ and $\frac18$ for the logarithm, respectively,
  and that the critical two-point function
  is asymptotic to a multiple of $|x|^{-2}$ as $|x| \to \infty$.


\subsection{Definition of the model}

For $d>0$, let $X$ denote the continuous-time simple random walk on $\Zd$.
That is, $X$ is the stochastic process
with right-continuous sample paths that takes its steps at the times
of the events of a rate-$2d$ Poisson process.  A step is independent both
of the Poisson process and of all other steps, and is taken uniformly
at random to one of the $2d$ nearest neighbours of the current
position.
The field of \emph{local times} $\lt_T = (\lt_T^x)_{x\in \Z^d}$
of  $X$, up to time $T \ge 0$,
is defined by
\begin{equation}
\label{e:LTx-def}
  \lt_T^x = \int_0^T \1_{X_t = x} \; dt
  .
\end{equation}
The \emph{self-intersection local time} and \emph{self-contact local time}
of $X$ up to time $T$ are the random variables defined, respectively, by
\begin{align}
\label{e:ITdef}
  I_T &=
  \sum_{x \in \Z^d} (\lt_T^x)^2
  = \int_0^T ds \int_0^T dt \; \1_{X_{s}=X_{t}}
  ,\\
\lbeq{CTdef}
  C_T
  &=
  \sum_{x \in \Z^d}\sum_{e\in\Ucal} \lt_T^x\lt_T^{x+e}
  = \int_0^T ds \int_0^T dt \; \1_{X_{s} \sim X_{t}}
  ,
\end{align}
where $\Ucal$ is the set of unit vectors in $\Zd$
and $y\sim x$ indicates that $x$ and $y$ are nearest neighbours.

Given $\beta > 0$ and $\gamma \in \R$,
we define
\begin{equation}
\label{e:Udef-neg}
U_{\beta,\gamma}(f)
=
\beta \sum_{x\in\Zd} f_x^2
- \frac{\gamma}{2d}
\sum_{x\in\Zd} \sum_{e\in\Ucal} f_x f_{x+e}
\end{equation}
for $f:\Zd\to \R$ with $f_x = 0$
for all but finitely many $x$.
The potential that associates an energy to $X$ in terms of its
field of local times is given by
\begin{equation}
  \label{e:V}
  U_{\beta,\gamma,T}
  =
  U_{\beta,\gamma}(L_T)
  =
  \beta I_T
  - \frac{\gamma}{2d}
  C_T
  .
\end{equation}
The energy $U_{\beta,\gamma,T}$ increases with the self-intersection local time,
corresponding to weak self-avoidance.  For $\gamma >0$, the energy decreases
when the self-contact local time increases, corresponding to a contact self-attraction.
For $\gamma<0$, the contact term is repulsive.  We are primarily interested in
the case of positive $\gamma$, but our results hold also for small negative $\gamma$.

Figure~\ref{fig:polymer-contact} shows a sample path $X$
and indicates one self-intersection and four self-contacts.
Although $I_T$ also receives contributions from the
time the walk spends at each vertex, and $C_T$ receives a contribution from each step,
these contributions have the same distribution for all walks taking the same number
of steps.  The depicted intersections and contacts are the meaningful ones.

\begin{figure}[ht]
 \centering\input{polymer-contact.pspdftex}
 \caption{Polymer with one self-intersection and four self-contacts shown.}
 \label{fig:polymer-contact}
\end{figure}

Let $a,b \in \Zd$, and
let $E_a$ denote the expectation for the
process $X$ started at $X(0)=a$.
We define
\begin{equation}
\label{e:c}
    c_T = E_a\left(e^{-U_{\beta,\gamma,T}}\right),
    \quad
    c_T(a,b) = E_a\left(e^{-U_{\beta,\gamma,T}}\1_{X_T = b}\right).
\end{equation}
By translation-invariance, $c_T$ does not depend on $a$.
For $\nu \in \R$, the \emph{two-point function} is defined by
\begin{align}
\lbeq{Gsa}
    G_{\beta,\gamma,\nu}(a,b) &=
    \int_0^\infty c_T(a,b) e^{-\nu T} \; dT,
\end{align}
and the \emph{susceptibility} is defined by
\begin{equation}
\label{e:suscept-def}
    \chi(\beta, \gamma, \nu)
    = \int_0^\infty c_T e^{-\nu T} \; dT
    = \sum_{x\in\Zd} G_{\beta,\gamma,\nu}(0, x)
    .
\end{equation}
For $p>0$, we define the \emph{correlation length of order $p$} by
\begin{equation}
\lbeq{xip-def}
    \xi_p(\beta,\gamma,\nu) = \left(\frac{1}{\chi(\beta, \gamma, \nu)}
    \sum_{x\in\Zd} |x|^p G_{\beta,\gamma,\nu}(0, x)
    \right)^{1/p}.
\end{equation}
In \eqref{e:Gsa}--\eqref{e:xip-def},
self-intersections are suppressed by the factor
$\exp[-\beta I_T]$, whereas nearest-neighbour
contacts are encouraged by the factor
$\exp[\frac{\gamma}{2d}C_T]$ when $\gamma > 0$.


\subsection{The critical point}

The right-hand sides of \eqref{e:Gsa}--\eqref{e:suscept-def}
are positive or $+\infty$,
and
monotone decreasing in $\nu$ by definition.
We define the \emph{critical point}
\begin{equation}
\label{e:nuc-def}
\nu_c(\beta, \gamma) = \inf \{ \nu \in \R : \chi(\beta, \gamma, \nu) < \infty \} .
\end{equation}
For $\gamma=0$, an elementary argument
shows that $\nu_c(\beta,0) > -\infty$ for all dimensions, and furthermore
that $\nu_c(\beta, 0) \in [ -2  \beta(-\Delta_{\Zd}^{-1})_{0,0}, 0]$ for dimensions $d>2$;
see \cite[Lemma~\ref{log-lem:csub}]{BBS-saw4-log}.
Here, $\Delta_{\Zd}$ is the Laplacian on $\Zd$, i.e., the $\Zd \times \Zd$
matrix with entries
\begin{equation}
\label{e:Deltaxy}
(\Delta_{\Zd})_{x, y} = \1_{x\sim y} - 2 d \1_{x=y}.
\end{equation}
An equivalent definition is as follows:
given a unit vector $e \in \Zd$, the discrete gradient is
defined by $\nabla^e f_x = f_{x+e}-f_x$, and the Laplacian is $\Delta_{\Zd}
f_{x} = \sum_{e \in \Ucal} \nabla^e f_x =
-\frac{1}{2}\sum_{e \in \Ucal}\nabla^{-e} \nabla^{e} f_x$.

To estimate the critical point when $\gamma \neq 0$,
we also define
\begin{align} \label{e:nabladef}
    |\nabla f_x|^2 &= \sum_{e\in\Ucal}
    |\nabla^e f_x|^2,
    \quad
    |\nabla f|^2 = \sum_{x\in\Zd} |\nabla f_x|^2.
\end{align}
From the definition, we see that
\begin{equation}
\label{e:sbp}
\sum_{x\in\Zd}   f_x \Delta_{\Zd} f_x
=
-\frac{1}{2} |\nabla f|^2.
\end{equation}
It follows that
\begin{equation}
\sum_{x\in\Zd} \sum_{e\in\Ucal} f_x f_{x+e}
=
2 d \sum_{x\in\Zd} f_x^2
+ \sum_{x\in\Zd} f_x \Delta_{\Zd} f_x
=
2 d \sum_{x\in\Zd} f_x^2
- \frac{1}{2} \sum_{x\in\Zd} |\nabla f_x|^2
\end{equation}
and so we get the useful representation:
\begin{equation}
\label{e:Udef-pos}
U_{\beta,\gamma}(f)
= (\beta - \gamma) \sum_{x\in\Zd} f_x^2
+ \frac{\gamma}{4d} \sum_{x\in\Zd} \sum_{e\in\Ucal} |\nabla^e f_x|^2.
\end{equation}
In particular,
\begin{equation}
  \label{e:V2}
  U_{\beta,\gamma,T} =
  (\beta - \gamma) I_T
  + \frac{\gamma}{4d}
  |\nabla \lt_T|^2
  .
\end{equation}
A version of \refeq{V2} can be found in \cite{HK01a}.

\begin{lemma}
\label{lem:nuc}
Let $d >0$.
Let $\beta>0$ and $|\gamma| < \beta$.
If $\gamma \ge 0$ then $\nu_c(\beta, \gamma) \in [\nu_c(\beta, 0),\nu_c(\beta-\gamma, 0)]$.
If $\gamma < 0$ then $\nu_c(\beta,\gamma) \in [\nu_c(\beta-\gamma,0),\nu_c(\beta,0)]$.
\end{lemma}

\begin{proof}
Suppose first that $\gamma \in [0,\beta)$.
It follows from \refeq{V} and \refeq{V2} that
\begin{equation}
    U_{\beta-\gamma,0,T} \le U_{\beta,\gamma,T} \le  U_{\beta,0,T},
\end{equation}
which implies the desired estimates for $\nu_c(\beta,\gamma)$.

On the other hand,
if $\gamma \in (-\beta, 0)$ then the inequalities are reversed and now
\begin{equation}
    U_{\beta,0,T} \le U_{\beta,\gamma,T} \le  U_{\beta-\gamma,0,T},
\end{equation}
which again implies the desired result.
\end{proof}


\subsection{The main result}

Our main result is the following theorem.
It shows that in dimension $d = 4$,
for sufficiently small $\beta$ and $\gamma$, the two-point function \refeq{Gsa} has the
same asymptotic decay, to leading order, as the simple random walk two-point function.
It also shows that the susceptibility and correlation length of order $p$
exhibit logarithmic corrections to mean-field behaviour.
These results were all proved for $\gamma=0$ in \cite{BBS-saw4,BBS-saw4-log,BSTW-clp},
and we extend them here to small nonzero $\gamma$.

We denote the Laplacian on $\Rd$ by $\Delta_{\Rd}$
and define a constant ${\sf c}_p$ by
\begin{equation}
{\sf c}_p^p = \int_{\R^4} |x|^p (-\Delta_{\R^4} + 1)^{-1}_{0x} \; dx.
\end{equation}

\begin{theorem} \label{thm:suscept}
  Let $d = 4$.
  There exist $\beta_* > 0$
  and a positive function $\gamma_* : (0, \beta_*) \to \R$
  such that whenever $0 < \beta < \beta_*$ and $|\gamma| < \gamma_*(\beta)$,
  there are constants $A_{\beta,\gamma}$ and $B_{\beta,\gamma}$ such that the following hold:

  \smallskip\noindent
  (i)
  The critical two-point function decays as
  \begin{equation}
    G_{\beta,\gamma,\nu_c}(0, x)
        =
    A_{\beta,\gamma} |x|^{-2} \left(1 + O\left(\frac{1}{\log |x|}\right)\right)
        \quad
    \text{as $|x|\to\infty$},
  \end{equation}
  with $A_{\beta,\gamma} = \frac{1}{4 \pi^2} (1 + O(\beta))$ as $\beta \downarrow 0$.

  \smallskip\noindent
  (ii)
  The susceptibility diverges as
  \begin{equation} \label{e:chieps-asympt}
    \chi(\beta, \gamma, \nu_c + \varepsilon)
    \sim B_{\beta,\gamma} \varepsilon^{-1} (\log \varepsilon^{-1})^{1/4},
    \quad \varepsilon\downarrow 0,
  \end{equation}
  with $B_{\beta,\gamma} = (\frac{\beta}{2\pi^2})^{1/4} (1 + O(\beta))$ as $\beta \downarrow 0$.

  \smallskip\noindent
  (iii) For any $p >0$,
  if $\beta_*$ is chosen small depending on $p$, then
  the correlation length of order $p$ diverges as
  \begin{equation} \label{e:xieps-asympt}
    \xi_p(\beta, \gamma, \nu_c + \varepsilon)
     \sim B_{\beta,\gamma}^{1/2} {\sf c}_p \varepsilon^{-1/2} (\log \varepsilon^{-1})^{1/8},
     \quad \varepsilon\downarrow 0.
  \end{equation}
\end{theorem}

Our method of proof extends the renormalisation group argument, used for $\gamma=0$  in
\cite{BBS-saw4,BBS-saw4-log,BSTW-clp,ST-phi4}, to small nonzero $\gamma$.
In Section~\ref{sec:finvol}, as a first step, we show that
the two-point function can be approximated by a
finite-volume one.  The finite-volume two-point function has a supersymmetric
integral representation
\cite{BM91,BEI92,BIS09}, which we state in Section~\ref{sec:intrep}.  These
two sections
do not involve the renormalisation group.
The application of the renormalisation group method requires the following new ingredients:
(i)
In Section~\ref{sec:K0bd}, we
provide estimates on the contact attraction which show that it
is compatible with the renormalisation group method developed in \cite{BS-rg-IE,BS-rg-step},
and also with the dynamical systems theorem proved in \cite{BBS-rg-flow}.
(ii)
In Section~\ref{sec:nucident},
we use the implicit function theorem to
extend the identification of the critical point from $\gamma=0$ to $\gamma\neq 0$, and
complete the proof of Theorem~\ref{thm:suscept}.

In fact, we demonstrate that after the introduction of $\gamma$,
chosen sufficiently small depending on $g$,
we may use the the same renormalisation group flow of the remaining coupling constants
as in the case $\gamma=0$,
to second order in these coupling constants.
Thus, since the critical exponents are determined by this second-order flow,
they are independent of small $\gamma$, and take the same values as for $\gamma=0$.
The critical value $\nu_c(\beta,\gamma)$ does, however, depend on $\gamma$.


\subsection{Critical exponents and polymer collapse}

\begin{figure}[ht]
  \scalebox{0.9}
  {\input{polymer-phasediagram.pspdftex}}
  \centering
  \caption{The predicted phase diagram for $d \ge 2$.
  \label{fig:phasediagram}}
\end{figure}
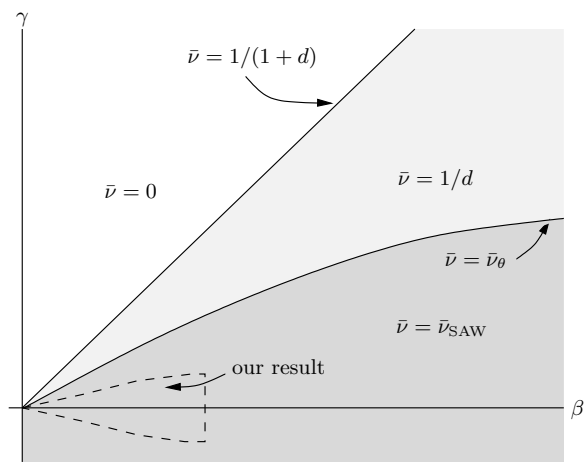

It has been known for decades that
self-avoiding walk obeys mean-field behaviour in dimensions $d \ge 5$.
In particular, a version of
Theorem~\ref{thm:suscept}
for the strictly self-avoiding walk (in discrete time with
$\beta=\infty$ and $\gamma=0$)
in dimensions $d \ge 5$ was proved in \cite{HS92a,Hara08} using the lace expansion \cite{BS85}.
In its original applications, the lace expansion relied on the
purely repulsive nature of the self-avoidance
interaction.  Models incorporating attraction require new ideas.
For a particular model with self-attraction and
specially chosen exponentially decaying step weights, the lace expansion was
used in \cite{Uelt02} to prove
that, for $d \ge 5$, the mean-square displacement grows diffusively for small attraction.
More recently \cite{Helm16}, the lace expansion has been applied in situations where repulsion occurs
only in an average sense.  In a further development \cite{HH17},
the lace expansion has been applied to a model of
strictly self-avoiding walk with a self-attraction that rewards visits to adjacent parallel
edges, to prove that sufficiently weak self-attraction does not affect the critical behaviour
in dimensions $d \ge 5$.  The results of \cite{HH17,Uelt02} for $d \ge 5$ complement our results
for $d=4$, via entirely different methods.

Assuming it exists, the critical exponent $\nubar$
for the mean-square displacement is defined by
\begin{equation}
\lbeq{msd}
    \langle |X(T)|^2 \rangle  =
    \frac{1}{c_T} E_0(|X(T)|^2 e^{-U_{\beta,\gamma,T}})
    \approx T^{2\nubar}
    ,
\end{equation}
possibly with logarithmic corrections.
A general tenet of the theory of critical phenomena asserts that other
natural length scales, such as the correlation length of order $p$,
are also governed by the exponent $\nubar$.
A typical argument for this, found in physics textbooks, goes as follows.
It is predicted that $c_T \approx e^{\nu_c T}T^{\bar\gamma -1}$, where $\bar\gamma$
is the critical exponent for the susceptibility (for $d=4$, $\bar\gamma=1$ with a logarithmic
correction, by \refeq{chieps-asympt}).
By definition,
\begin{equation}
\label{e:xi2}
\xi_2(\beta, \gamma, \nu)^2
  =
\frac
  {\int_0^\infty \langle |X(T)|^2 \rangle c_T e^{-\nu T} \; dT}
  {\int_0^\infty c_T e^{-\nu T} \; dT}
.
\end{equation}
In \refeq{xi2}, we substitute the asymptotic formula for $c_T$,
as well as \refeq{msd}, to obtain
\begin{equation}
\label{e:xi-nubar}
    \xi_2(\beta, \gamma, \nu)
    \approx
    (\nu - \nu_c)^{-\nubar} \quad \text{as}\; \nu \downarrow \nu_c,
\end{equation}
with the same exponent $\nubar$ as in \refeq{msd}.

The weakly self-avoiding walk with contact self-attraction is a model for polymer collapse.
Polymer collapse corresponds to a
discontinuous reduction in the exponent $\nubar$ as $\gamma$ increases.
A summary of results, predictions, and references can be found in
\cite[Chapter~6]{Holl09}.
See also \cite{Jans15,Vand98}.
The predicted phase diagram for dimensions $d \ge 2$ is shown
in Figure~\ref{fig:phasediagram}.
The predicted values of the exponent at the $\theta$-transition are
$\nubar_\theta = \frac 47$ for $d=2$ and $\nubar_\theta = \frac 12$ for $d \ge 3$ \cite{Holl09}.
The phase labelled $\nubar_{\rm SAW}$ takes its name from the fact that in this
phase the model with attraction is predicted to be in the same universality class
as the self-avoiding walk.
The predicted values of the exponent $\nubar_{\rm SAW}$ for the
self-avoiding walk are respectively
$\frac 34$, $0.587597(7)$, $\frac 12$ for $d=2,3,4$ (with a logarithmic correction
for $d=4$; see \cite{Clis10} for $d=3$), and it has been proved that $\nubar_{\rm SAW}=\frac 12$
for $d \ge 5$ \cite{BS85,HS92a}.
It remains a major challenge in the mathematical theory of polymers
to prove the full validity of the phase diagram in all dimensions $d\geq 2$.
Very recently, the existence of a collapse transition
(a singularity of the free energy) has been proven for a
2-dimensional \emph{prudent} self-avoiding walk with contact self-attraction \cite{PT16}.

For $\gamma \geq 0$,
the significance of the restriction $\gamma <\beta$
has been noted for a closely related discrete-time model,
for which it is proved that for $\gamma > \beta$
the walk is in a compact phase in the sense that $\nubar = 0$, whereas for $\gamma < \beta$
it is the case that $\nubar  \ge 1/d$ \cite{HK01a}.
In the compact
phase, the discrete-time model obeys the analogue of
$c_T \approx e^{kT^2}$ with $k>0$, so $\chi(\beta,\gamma,\nu)=\infty$ for all $\nu \in \R$
and $\nu_c = +\infty$.
For the 1-dimensional case, the behaviour for the transition line $\gamma=\beta$ has been
studied in \cite{HKK02}.

The axis $\gamma=0$ corresponds to the weakly self-avoiding walk which is
well understood in dimensions $d \ge 5$ \cite{BS85,HS92a}, and
in dimension $4$ \cite{BBS-saw4,BBS-saw4-log,BSTW-clp}.
Theorem~\ref{thm:suscept} extends the results of \cite{BBS-saw4,BBS-saw4-log,BSTW-clp}
for dimension $d = 4$ to the region bounded by the dashed line.
Our results show that for $d=4$ there is no polymer collapse for small contact self-attraction,
in the sense that the critical behaviour remains the same with small contact attraction
as with no contact attraction.
In particular, Theorem~\ref{thm:suscept}(iii) shows that,
in the sense of \eqref{e:xi-nubar}, when $\gamma$ is small,
$\nubar = \frac{1}{2}$ holds with a logarithmic correction.


\section{Finite-volume approximation}
\label{sec:finvol}

The first step in the proof of Theorem~\ref{thm:suscept}
is an approximation of $G_{\beta,\gamma,\nu}(a,b)$ and $\chi(\beta, \gamma, \nu)$
by finite-volume analogues of these quantities.
This is the content of Proposition~\ref{prop:finvol}.

Before proving the proposition, we require some preliminaries.
Let $P^n$ be the projection
of $\Zd$ onto the discrete torus of side $n$,
which we denote $\Z_n^d$.
Then $P^n$ has a natural action
on the path space $(\Zd)^{[0,\infty)}$. We let
$X^n = P^n(X)$ be the projection of $X$
and note that $X^n$ is a simple random walk on $\Z^d_n$.

We call $h = (h_x)_{x\in\Zd}$ a \emph{field of path functionals} if
$h_x : (\Zd)^{[0,\infty)} \to \R$ is a function on continuous-time paths
for each $x \in \Zd$;
a simple example is given by the local time functional.
We assume that the \emph{random} field $h(X) = (h_x(X))_{x\in\Zd}$
has finite support almost surely, i.e.,
with probability $1$, $h_x(X) = 0$ for all but finitely many $x$.
Denote by $h(X^n)$ the corresponding random field for $X^n$, i.e., for $x \in \Z_n^d$,
\begin{equation}
h_x(X^n) = \sum_{y\in\Zd} h_{x+ny}(X).
\end{equation}

Given a positive integer $k$, we define
$Q_k \subset \Z^d$ by $Q_k = \{y \in \Z^d : 0 \le y_i < k, \;   i=1,\ldots,d\}$.
Then, for integers $n,k \ge 1$,
\begin{equation}
\label{e:ffold1}
    \sum_{y \in Q_k} h_{x+ny}(X^{kn})
  = \sum_{y \in Q_k} \sum_{z\in\Zd} h_{x+ny+knz}(X)
  = \sum_{y\in\Zd} h_{x+ny}(X)
  = h_x(X^n),
\end{equation}
and it follows by summation over $x \in \Z^d_n$ that
\begin{equation}
\label{e:ffold2}
\sum_{x\in\Z^d_{kn}} h_x(X^{kn})
  =
\sum_{x\in\Z^d_n} h_x(X^n).
\end{equation}

\begin{lemma}
\label{lem:mono}
Let $n,k \ge 1$ and let $f$ and $g$ be nonnegative fields of path functionals
with finite support almost surely.
Then
\begin{equation}
\sum_{x\in\Z^d_{kn}} f_x(X^{kn}) g_x(X^{kn})
  \leq
\sum_{x\in\Z^d_n} f_x(X^n) g_x(X^n).
\end{equation}
\end{lemma}

\begin{proof}
By \eqref{e:ffold2} and \eqref{e:ffold1},
\begin{equation}
\sum_{x\in\Z_{kn}^d} f_x(X^{kn}) g_x(X^{kn})
  =
\sum_{x\in\Z_n^d}
\sum_{y \in Q_k}
  f_{x+ny}(X^{kn}) g_{x+ny}(X^{kn}).
\lbeq{mono}
\end{equation}
By nonnegativity and two more applications of \eqref{e:ffold1},
\begin{align}
\sum_{x\in\Z_n^d}
\sum_{y \in Q_k}
f_{x+ny}(X^{kn}) g_{x+ny}(X^{kn})
  &\le \sum_{x\in\Z_n^d}
      \left(\sum_{y \in Q_k} f_{x+ny}(X^{kn})\right)
      \left(\sum_{y \in Q_k} g_{x+ny}(X^{kn})\right) \nonumber \\
  &= \sum_{x\in\Z_n^d} f_x(X^n) g_x(X^n).
\end{align}
This completes the proof.
\end{proof}

Fix $L \geq 2$ and $N \geq 1$.
We write $\Lambda_N$ for the torus $\Z^d_n$ with $n=L^N$.
Thus, $X^{L^N}$ is the simple random walk on $\Lambda_N$.
For $F_T = F_T(X)$ any one of the functions $L_T^x,I_T,C_T$
of $X$ defined in \eqref{e:LTx-def}--\eqref{e:CTdef},
we write $F_{N,T} = F_T(X^{L^N})$. For instance, with $n=L^N$,
\begin{equation}
    L^x_{N,T} = \int_0^T \1_{X^{n}_t=\;x} \; dt,
    \quad I_{N,T} = \sum_{x \in \Lambda_N}(L_{N,T}^x)^2 .
\end{equation}
We apply Lemma~\ref{lem:mono} with $k = L$ and $n = L^N$
for three choices of $f,g$:
\begin{alignat}{2}
\label{e:ILT-mon}
I_{N+1,T} &\leq I_{N,T}\quad &&(f_x=g_x=L_T^x),
\\
\label{e:CSA-mon}
C_{N+1,T} &\leq C_{N,T} \quad &&(f_x=\textstyle{\sum_{e\in \Ucal}L_T^{x+e}},\; g_x=L_T^x) ,
\\
\sum_{x\in\Lambda_{N+1}} |\nabla^e L^x_{N+1,T}|^2
  &\leq
\sum_{x\in\Lambda_N} |\nabla^e L^x_{N,T}|^2
\quad &&
(f_x = g_x = \left|\nabla^e L_T^x\right|).
\lbeq{nabL}
\end{alignat}
Summation of \refeq{nabL} over $e\in\Ucal$ also gives
\begin{align}
\label{e:gradLT-mon}
\sum_{x\in\Lambda_{N+1}} |\nabla L^x_{N+1,T}|^2
  \leq
\sum_{x\in\Lambda_N} |\nabla L^x_{N,T}|^2.
\end{align}

We identify the vertices of $\Lambda_N$ with nested subsets of $\Zd$,
centred at the origin (approximately if $L$ is even),
with $\Lambda_{N+1}$ paved by $L^d$ translates of $\Lambda_N$.
We can thus define $\partial \Lambda_N$ to be the inner vertex boundary of $\Lambda_N$.
We denote the expectation of $X^{L^N}$ started from $a \in \Lambda_N$ by $E^{\Lambda_N}_a$
and define
\begin{align}
\label{e:cN}
c_{N,T}(a, b)
    &= E^{\Lambda_N}_a \left( e^{-U_{\beta,\gamma,T}} \1_{X(T)=b} \right)
    \quad (a, b \in \Lambda_N), \\
c_{N,T}
    &= E^{\Lambda_N}_0 \left( e^{-U_{\beta,\gamma,T}} \right).
\end{align}
The finite-volume two-point function and susceptibility
are defined by
\begin{align}
G_{N,\beta,\gamma,\nu}(a,b)
    &= \int_0^\infty c_{N,T}(a, b) e^{-\nu T} \; dT, \\
\chi_N(\beta, \gamma, \nu)
    &= \int_0^\infty c_{N,T} e^{-\nu T} \; dT
    .
    \label{e:chiNdef}
\end{align}

\begin{prop}
\label{prop:finvol}
Let $d >0$, $\beta >0$ and $\gamma < \beta$. For all $\nu \in \R$,
\begin{equation}
\label{e:Givlc}
\lim_{N \to \infty}
G_{N,\beta,\gamma,\nu}(a,b)
=
G_{\beta,\gamma,\nu}(a,b)
\end{equation}
and
\begin{equation}
\label{e:chilim}
\lim_{N\to\infty}\chi_N(\beta,\gamma,\nu)=   \chi(\beta,\gamma,\nu).
\end{equation}
\end{prop}

\begin{proof}
Fix $a, b \in \Zd$, and consider $N$ sufficiently large that $a,b$ can be identified
with points in $\Lambda_N$.
By \eqref{e:V2}, \eqref{e:ILT-mon} and \eqref{e:gradLT-mon} (if $0 \le \gamma <\beta$),
or by \eqref{e:V}, \eqref{e:ILT-mon} and \eqref{e:CSA-mon} (if $\gamma < 0$),
\begin{equation}
\label{e:ctmon}
c_{N,T}(a, b) \leq c_{N+1,T}(a, b).
\end{equation}
Thus, \eqref{e:Givlc} follows by monotone convergence, once we show that
\begin{equation}
\lim_{N\to\infty} c_{N,T}(a, b) = c_T(a, b).
\end{equation}

This follows as in \cite[(2.8)]{BBS-saw4}.
That is, first we define
\begin{align}
c_{N,T}^*(a, b)
  &=
E^{\Lambda_N}_a
\left(
  e^{-U_{\beta,\gamma,T}} \1_{X(T)=b} \1_{\{X([0, T]) \cap \partial \Lambda_N \neq \varnothing\}}
\right) \\
c_T^*(a, b)
  &=
E_a
\left(
  e^{-U_{\beta,\gamma,T}} \1_{X(T)=b} \1_{\{X([0, T]) \cap \partial \Lambda_N \neq \varnothing\}}
\right).
\end{align}
Since walks which do not reach $\partial \Lambda_N$ make equal contributions to both
$c_T(a,b)$ and $c_{N,T}(a,b)$,
we have
\begin{equation}
c_T(a, b) - c_T^*(a, b) = c_{N,T}(a, b) - c_{N,T}^*(a, b).
\end{equation}
Thus,
\begin{align}
|c_T(a, b) - c_{N,T}(a, b)|
= |c_T^*(a, b) - c_{N,T}^*(a, b)|
\leq c_T^*(a, b) + c_{N,T}^*(a, b).
\end{align}
Let $P^{\Lambda_N}_a$ and $P_a$ be the measures
associated with $E^{\Lambda_N}_a$ and $E_a$, respectively.
With $Y_t$ a rate-$2d$ Poisson process with measure ${\sf P}$,
\begin{align}
 c_T^*(a, b) + c_{N,T}^*(a, b)
  &\leq P_a (X([0, T]) \cap \partial\Lambda_N \neq \varnothing)
    + P^{\Lambda_N}_a (X([0, T]) \cap \partial\Lambda_N \neq \varnothing) \nonumber \\
  &\leq 2 {\sf P} (Y_T \geq \diam{\Lambda_N}) \to 0
\end{align}
as $N\to\infty$.
This completes the proof of \eqref{e:Givlc}.

Finally, by monotone convergence of $G_N$ to $G$,
for $\nu \in \R$,
\begin{equation}
\lim_{N\to\infty} \chi_N(g, \gamma, \nu)
    = \sum_{b\in\Zd} \lim_{N\to\infty} G_{N,g,\gamma,\nu}(a, b) \1_{b\in\Lambda_N}
    = \chi(g, \gamma, \nu),
\end{equation}
which proves \eqref{e:chilim}.
\end{proof}


\section{Integral representation and progressive integration}
\label{sec:intrep}

In this section, we reformulate the model in terms of a perturbation of a supersymmetric
Gaussian integral, in order to prepare for the application of the renormalisation group.
The integral representation, which is a special case of a result from \cite{BIS09},
makes use of the Grassmann integral. We begin by recalling the definition of the Grassmann
integral in Section~\ref{sec:forms} and state the integral representation in Section~\ref{sec:Gintrep}.
In Section~\ref{sec:Gauss-approx}, we split the integral into a Gaussian part and a perturbation.
The basic idea underlying the renormalisation group is the progressive evaluation of this
Gaussian integral via a multi-scale decomposition of its covariance, which we introduce in Section~\ref{sec:prog}.


\subsection{Boson and fermion fields}
\label{sec:forms}

We fix $N$ and write $\Lambda = \Lambda_N$.
Given complex variables $\phi_x, \bar\phi_x$
(the boson field) for $x \in \Lambda$,
we define the differentials (the fermion field)
\begin{equation}
\psi_x = \frac{1}{\sqrt{2\pi i}} d\phi_x,
\quad
\bar\psi_x = \frac{1}{\sqrt{2\pi i}} d\bar\phi_x,
\end{equation}
where we fix a choice of complex square root.
The fermion fields are multiplied with each other
via the anti-commutative wedge product,
though we suppress this in our notation.

A differential form that is the
product of a function of $(\phi, \bar\phi)$
with $p$ differentials is said to have degree $p$.
A sum of forms of even degree is said to be \emph{even}.
We introduce a copy $\bar\Lambda$ of $\Lambda$
and we denote the copy of $X \subset \Lambda$ by $\bar X \subset \bar\Lambda$.
We also denote the copy of $x \in \Lambda$
by $\bar x \in \bar\Lambda$ and define $\phi_{\bar x} = \bar\phi_x$ and $\psi_{\bar x} = \bar\psi_x$.
Then any differential form $F$ can be written
\begin{equation}
\lbeq{FinNcal}
F
=
\sum_{\vec y}
F_{\vec y} (\phi, \bar\phi)
\psi^{\vec y}
\end{equation}
where the sum is over finite sequences $\vec y$ over $\Lambda\sqcup\bar\Lambda$,
and $\psi^{\vec y} = \psi_{y_1} \ldots \psi_{y_p}$ when $\vec y = (y_1, \ldots, y_p)$.
When $\vec y = \varnothing$ is the empty sequence,
$F_\varnothing$ denotes the $0$-degree (bosonic) part of $F$.

In order to apply the results of \cite{BBS-saw4-log,BBS-saw4,BSTW-clp}, we require
smoothness of the coefficients $F_{\vec y}$ of $F$.  For Theorem~\ref{thm:suscept}(i,ii),
we need these coefficients to be $C^{10}$, and for Theorem~\ref{thm:suscept}(iii) we require
a $p$-dependent number of derivatives for the analysis of $\xi_p$, as discussed in \cite{BSTW-clp}.
We let $\Ncal$ be the algebra of even forms with sufficiently smooth coefficients
and we let $\Ncal(X) \subset \Ncal$ be the sub-algebra of even forms only depending on fields
in $X$. Thus, for $F \in \Ncal(X)$, the sum in \eqref{e:FinNcal} runs over sequences $\vec y$
over $X \sqcup \bar X$.
Note that $\Ncal = \Ncal(\Lambda)$.

Now let $F = (F_j)_{j \in J}$ be a finite collection of even forms
indexed by a set $J$
and write $F_\varnothing = (F_{\varnothing,j})_{j \in J}$.
Given a $C^\infty$ function $f : \R^J \to \C$, we define
$f(F)$ by its Taylor expansion about $F_\varnothing$:
\begin{equation}
f(F) = \sum_\alpha \frac{1}{\alpha!} f^{(\alpha)}(F_\varnothing) (F - F_\varnothing)^\alpha.
\end{equation}
The summation terminates as a finite sum,
since $\psi_x^2 = \bar\psi_x^2 = 0$ due to the anti-commut\-ative product.

We define the integral
$\int F$
of a differential form $F$ in the usual way
as the Riemann integral of its top-degree part
(which may be regarded as a function
of the boson field).
In particular, given a positive-definite
$\Lambda \times \Lambda$ symmetric matrix $C$
with inverse $A = C^{-1}$,
we define the \emph{Gaussian expectation}
(or \emph{super-expectation}) of $F$ by
\begin{equation}
\lbeq{ExCF}
\Ex_C F = \int e^{-S_A} F,
\end{equation}
where
\begin{equation}
\label{e:action}
S_A = \sum_{x\in\Lambda} \Big(\phi_x (A\bar\phi)_x + \psi_x (A \bar\psi)_x\Big).
\end{equation}

Finally, for $F = f(\phi, \bar\phi) \psi^{\vec y}$,
we let
\begin{equation}
\theta F = f(\phi + \xi, \bar\phi + \bar\xi) (\psi + \eta)^{\vec y},
\end{equation}
where $\xi$ is a new boson field, $\eta = (2\pi i)^{-1/2} d\xi$ a new fermion field,
and $\bar\xi, \bar\eta$ are the corresponding conjugate fields.
We extend $\theta$ to all $F \in \Ncal$ by linearity
and define the convolution operator $\Ex_C\theta$ by letting
$\Ex_C\theta F \in \Ncal$ denote the Gaussian expectation of $\theta F$ with respect
to $(\xi, \bar\xi, \eta, \bar\eta)$, with $\phi,\phib,\psi,\psib$ held fixed.


\subsection{Integral representation of the two-point function}
\label{sec:Gintrep}

An integral representation formula applying to general local time functionals
is given in \cite{BEI92,BIS09}; see also \cite[Appendix~A]{ST-phi4}.
We state the result we need in the proposition below.

Let $\Delta$ denote the Laplacian on $\Lambda$,
i.e.\ $\Delta_{xy}$ is given by the right-hand side of
\eqref{e:Deltaxy} for $x, y \in \Lambda$.
We define the differential forms:
\begin{align}
\tau_x
&= \phi_x \bar\phi_x
+ \psi_x   \bar\psi_x
\\
\label{e:addDelta}
\tau_{\Delta,x}
&=
\frac 12 \Big(
\phi_{x} (- \Delta \bar{\phi})_{x} + (- \Delta \phi)_{x} \bar{\phi}_{x} +
\psi_{x}  (- \Delta \bar{\psi})_{x} + (- \Delta \psi)_{x}  \bar{\psi}_{x}
\Big) \\
|\nabla \tau_x|^2
&= \sum_{e\in\Ucal} (\nabla^e \tau)_x^2.
\end{align}

\begin{prop}
Let $d > 0$ and $\beta > 0$. For $\gamma < \beta$ and $\nu \in \R$,
\begin{align}
G_{N,\beta,\gamma,\nu}(a, b)
&=  \int
    e^{-\sum_{x\in\Lambda}
    \left(
    U_{\beta,\gamma}(\tau)
    + \nu \tau_x + \tau_{\Delta,x}\right)} \bar\phi_a \phi_b
    .
    \label{e:Grep-pos-bis}
\end{align}
\end{prop}

\begin{proof}
The proof is identical to the proof of the $p = 1$ case of
\cite[Proposition~\ref{phi4-prop:Integral-Representation}]{ST-phi4}
when, in the notation used in \cite{ST-phi4}, we set
\begin{equation}
F(S) = e^{-U_{\beta,\gamma}(S) - (\nu - 1) \sum_{x\in\Lambda} S_x}
\end{equation}
in \cite[\eqref{phi4-e:intrep1}]{ST-phi4}.
\end{proof}

\subsection{Gaussian approximation}
\label{sec:Gauss-approx}

We divide
the integral in \eqref{e:Grep-pos-bis} into
a Gaussian part and a perturbation.  Although the division is arbitrary here,
a careful choice of the division must be made, and it is made in Theorem~\ref{thm:flow-flow}.
We require several definitions.
Let $z_0>-1$ and $m^2 >0$. We set
\begin{equation}
\label{e:gg0}
g_0 = (\beta - \gamma) (1 + z_0)^2,
\quad
\nu_0 = \nu (1 + z_0) - m^2,
\quad
\gamma_0 = \frac{1}{4d} \gamma (1 + z_0)^2,
\end{equation}
and define
\begin{equation}
\label{e:V0def}
  V^+_{0,x}
  = g_0\tau_x^2 + \nu_0 \tau_x + z_0 \tau_{\Delta,x},
  \quad
  U^+_x = |\nabla \tau_x|^2.
\end{equation}
The monomial $U^+_x$ should not be confused with
the potential $U_{\beta,\gamma}$.
We define
\begin{equation}
\label{e:Z0def}
  Z_0
  =
  \prod_{x\in \Lambda} e^{-(V^+_{0,x} + \gamma_0 U^+_x)},
\end{equation}
and, with $C = (-\Delta + m^2)^{-1}$ and with the expectation given by \refeq{ExCF},
\begin{equation}
\label{e:ZNdef}
Z_N = \Ex_C \theta Z_0.
\end{equation}

Recall that $Z_{N,\varnothing}$ denotes the $0$-degree part of $Z_N$.
We define a test function $\1: \Lambda_N \to \R$ by $\1_x=1$ for all $x$,
and write $D^2 Z_{N,\varnothing}(0, 0; \1, \1)$ for the directional derivative of
$Z_{N,\varnothing}$
at $(\phi, \bar\phi) = (0, 0)$, with both directions equal to $\1$.
That is,
\begin{equation}
D^2 Z_{N,\varnothing}(0, 0; \1, \1)
  =
\ddp{^2}{s\partial t} Z_{N,\varnothing}(s \1, t\1)\big|_{s=t=0}.
\end{equation}

\begin{prop}
Let $d > 0$, $\gamma, \nu \in \R$, $\beta >0$ and $\gamma <\beta$.
If the relations \eqref{e:gg0} hold, then
\label{prop:intrep}
\begin{equation}
\label{e:GG2}
G_{N,\beta,\gamma,\nu}(a,b)
    =
(1+z_0)
\Ex_C (Z_0 \bar\phi_a \phi_b),
\end{equation}
and
\begin{equation}
\label{e:chichibar}
  \chi_N\left(\beta,\gamma,\nu\right)
  = (1+z_0)\hat\chi_N(m^2, g_0, \gamma_0, \nu_0, z_0)
  ,
\end{equation}
with
\begin{equation}
  \label{e:chibarm}
  \hat\chi_N(m^2, g_0, \gamma_0, \nu_0, z_0)
  = \frac{1}{m^2}  + \frac{1}{m^4} \frac{1}{|\Lambda|} D^2 Z_{N,\varnothing}(0, 0; \1, \1).
\end{equation}
\end{prop}

\begin{proof}
We make the change of variables
$\varphi_x \mapsto (1 + z_0)^{1/2} \varphi_x$ (with $\varphi = \phi, \bar\phi, \psi, \bar\psi$)
in \eqref{e:Grep-pos-bis}, and obtain
\begin{align}
    G_{N,\beta,\gamma,\nu}(a, b)
    &=  (1+z_0)
    \int
    e^{-\sum_{x\in\Lambda}
    \left(
    g_0 \tau_x^2 + \gamma_0 |\nabla \tau_x|^2
    + \nu (1+z_0) \tau_x + (1+z_0)\tau_{\Delta,x}\right)} \bar\phi_a \phi_b
    .
    \label{e:Grep-pos}
\end{align}
Then, for any $m^2 \in\R$, we have
\begin{equation}
\lbeq{GNint}
G_{N,\beta,\gamma,\nu}(a, b)
    =
(1 + z_0) \int
e^{-\sum_{x\in\Lambda} (\tau_{\Delta,x} + m^2 \tau_x)}
Z_0 \bar\phi_a \phi_b
\end{equation}
($m^2$ simply cancels with $\nu_0$ on the right-hand side).
We use this with $m^2>0$, so that the inverse matrix $C=(-\Delta+m^2)^{-1}$ exists.
By symmetry of the matrix $\Delta$, \refeq{action} gives
\begin{equation}
\label{e:SAtauDelta}
S_{(-\Delta+m^2)}
=
\sum_{x\in\Lambda} \left( \tau_{\Delta,x}
+ m^2  \tau_x \right).
\end{equation}
Then \eqref{e:GG2} follows from \refeq{GNint}--\eqref{e:SAtauDelta} and \refeq{ExCF}.  Summation
over $b\in \Lambda_N$ gives the formula $\chi_N(\beta,\gamma,\nu) = (1+z_0)\sum_{x\in \Lambda} \Ex_C
(Z_0\phib_0\phi_x)$.  Then \refeq{chichibar}, with \refeq{chibarm}, follows
by an elementary computation as in \cite[Section~\ref{log-sec:ga}]{BBS-saw4-log}.
\end{proof}

\subsection{Progressive integration}
\label{sec:prog}

The identity \eqref{e:GG2} splits the two-point function into a Gaussian part
and a perturbation $Z_0$. The Gaussian part is parametrised by $(m^2, z_0)$,
although the dependence on $z_0$ has been shifted out
of the integral.
We analyse the integral \refeq{GG2} using the
renormalisation group method developed in
\cite{BS-rg-norm,BS-rg-loc,BBS-rg-pt,BS-rg-IE,BS-rg-step}, which is itself inspired by \cite{WK74}.
This method is based on a decomposition
\begin{equation}
C = C_1 + \cdots + C_{N-1} + C_{N,N},
\end{equation}
of the covariance $C$ used to define $Z_N$ in \eqref{e:ZNdef},
where $C_1, \ldots, C_{N-1}, C_{N,N}$ are covariances.
For simplicity, we write $C_N = C_{N,N}$.
A \emph{finite-range} decomposition of this sort was constructed in \cite{Baue13a,BGM04}.
Specifically, we use the decomposition of \cite{Baue13a}.

The covariance decomposition allows us to evaluate $Z_N$ progressively
by defining inductively
\begin{equation}
\label{e:Zjdef}
Z_{j+1} = \Ex_{C_{j+1}}\theta Z_j \quad (j < N).
\end{equation}
It is a basic fact that a sum of two independent
Gaussian random variables with covariances $C'$ and $C''$
is itself Gaussian with covariance $C' + C''$. By
\cite[Proposition~\ref{norm-prop:conv}]{BS-rg-norm}, this property
extends to the Gaussian super-expectation in the sense that
\begin{equation}
\Ex_C\theta = \Ex_{C_N}\theta \circ \ldots \circ \Ex_{C_1}\theta.
\end{equation}
Thus, the definition of $Z_{j+1}$ in \eqref{e:Zjdef} agrees
with \eqref{e:ZNdef} when ${j+1} = N$.

From the perspective of the renormalisation group,
we view the map $Z_j \mapsto Z_{j+1}$ as defining a dynamical system.
The evaluation of $Z_N$ can be accomplished by studying this system's
dependence on its initial condition, as we discuss in the next section.


\section{Initial coordinates for the renormalisation group}
\label{sec:K0bd}

Following the approach of \cite{BBS-saw4-log},
we represent $Z_j$
by a pair of coordinates $I_j$ and $K_j$ that capture the \emph{relevant} (expanding),
\emph{marginal},
and \emph{irrelevant} (contracting) parts of $Z_j$.
We begin in Section~\ref{sec:IK} by defining coordinates $(I_0, K_0)$ for $Z_0$.
Norms used to control the evolution of these coordinates are introduced
in Section~\ref{sec:norms}, and it is shown in
Sections~\ref{sec:K0bds}--\ref{sec:KWcal} that $K_0$ satisfies norm estimates
that permit the results of \cite{BS-rg-step,BBS-rg-flow} to be applied.
The initial coordinate $K_0$
depends on the coupling constants $(g_0, \gamma_0, \nu_0, z_0)$
of \eqref{e:gg0}, and regularity of $K_0$ as a function of these variables
is established in Section~\ref{sec:Ksmooth}.


\subsection{Initial coordinates for the renormalisation group}
\label{sec:IK}

We now divide $Z_0$ into coordinates $I_0$ and $K_0$.
The division depends on the sign of $\gamma$.

\subsubsection{Coordinates for positive \texorpdfstring{$\gamma$}{gamma}}
\label{sec:IKplus}

Assume that $\gamma \ge 0$.
For $X \subset \Lambda$, we define
\begin{equation}
  \label{e:IK0def}
    I_0^+(X) = \prod_{x\in X} e^{-V^+_{0,x}},
    \quad\quad
    K_0^+(X) = \prod_{x \in X} I_{0,x}^+ (e^{-\gamma_0 U^{+}_{x}} - 1).
\end{equation}
Here, $I^+_{0,x} = I^+_0(\{x\})$, and we usually denote evaluation at a singleton
by a subscript.
By definition and binomial expansion,
\begin{equation}
\label{e:Z00}
  Z_0
  =
  \prod_{x\in \Lambda} \left( I^+_{0,x} + K^+_{0,x} \right)
  =
  \sum_{X\subset \Lambda} I_0^+(\Lambda \setminus X) K_0^+(X)
  .
\end{equation}
This \emph{polymer gas representation} of $Z_0$ extends a much simpler representation
used to study the weakly self-avoiding walk previously, e.g., in \cite{BBS-saw4,BBS-saw4-log}.
In particular, when $\gamma_0 = 0$,
\begin{equation}
\lbeq{K0trivial}
K_0^+(X) = \1_\varnothing(X) =
\begin{cases}
1 & X = \varnothing \\
0 & \text{otherwise,}
\end{cases}
\end{equation}
and \eqref{e:Z00} agrees with
\cite[\eqref{log-e:Zinit}]{BBS-saw4-log}.
Thus the effect of nonzero $\gamma_0$ is incorporated entirely into the non-trivial
$K_0^+$ of \refeq{IK0def}, rather than \refeq{K0trivial}.

Then
$(V^+_0, K_0^+)$ can be viewed as the initial condition of the
dynamical system \refeq{Zjdef}.
This initial condition is \emph{not} uniquely defined as a function of
$(\beta, \gamma, \nu)$. Rather, the constraints \eqref{e:gg0} leave us with the freedom to choose
$\nu_0$ and $z_0$ as we please. The key to the success of the renormalisation group
method is the identification of \emph{critical} values $\nu_0^c, z_0^c$ that lie on
a stable manifold for the \emph{Gaussian fixed point} $(V_0, K_0) = 0$.
The existence of the stable manifold,
which is a highly non-trivial fact, is obtained using the main result of \cite{BBS-rg-flow}.
This result allows for the possibility that $K_0^+$ is non-zero as long as $\|K^+_0\| = O(g_0^3)$
in an appropriate norm.
We take advantage of this additional generality
in order to prove Theorem~\ref{thm:suscept}.


\subsubsection{Coordinates for negative \texorpdfstring{$\gamma$}{gamma}}

Assume that $\gamma <0$.
Define
\begin{equation}
V^-_{0,x}
    =
V^+_{0,x} + 4 d \gamma_0 \tau_x^2,
\quad
U^-_x = 2 \sum_{e\in\Ucal} \tau_x \tau_{x+e}.
\end{equation}
By the identity
\begin{equation}
\sum_{x\in\Lambda}
\Big(
  g_0 \tau_x^2 + \gamma_0 \sum_{e\in\Ucal} (\nabla^e \tau_x)^2
\Big)
  =
\sum_{x\in\Lambda}
\Big(
  (g_0 + 4d \gamma_0) \tau_x^2 - 2 \gamma_0 \sum_{e\in\Ucal} \tau_x \tau_{x+e}
\Big),
\end{equation}
we can write
\begin{equation}
Z_0 = \prod_{x\in\Lambda} (I^-_{0,x} +  K^-_{0,x}) = \sum_{X\subset\Lambda} I^-_0(\Lambda \setminus X) K^-_0(X),
\end{equation}
with
\begin{equation}
\label{e:IK0min}
I^-_0(X) = \prod_{x\in X} e^{-V^-_{0,x}},
\quad\quad
K^-_0(X) = \prod_{x\in X} I^-_{0,x} (e^{\gamma_0 U^-_x} - 1).
\end{equation}
Thus, we can parametrise $Z_0$ via either pair $(I^\pm_0, K^\pm_0)$.
We use $(I^+_0, K^+_0)$ when $\gamma_0 \geq 0$
and $(I^-_0, K^-_0)$ when $\gamma_0 < 0$.
With this convention,
\begin{equation}
\label{e:Kpm}
K^\pm_0(X) = \prod_{x\in X} I^\pm_{0,x} (e^{-|\gamma_0| U^\pm_x} - 1)
\quad \text{(use $+$ for $\gamma_0 \ge 0$, use $-$ for $\gamma_0<0$)}.
\end{equation}


\subsection{Norms}
\label{sec:norms}

In this section, we recall some definitions and basic facts concerning norms,
from \cite{BS-rg-norm}.
For now, we only consider the case of scale $j = 0$.

Recall the notation introduced in Section~\ref{sec:forms}.
A \emph{test function} $g$ is defined to be a function $(\vec x, \vec y) \mapsto g_{\vec x,\vec y}$,
where $\vec x$ and $\vec y$ are finite sequences of elements in $\Lambda \sqcup \bar\Lambda$.
When $\vec x$ or $\vec y$ is the empty sequence $\varnothing$,
we drop it from the notation as long as this causes no confusion;
e.g., we may write $g_{\vec x} = g_{\vec x,\varnothing}$.
The length of a sequence $\vec x$ is denoted $|\vec x|$.
Gradients of test functions are defined component-wise.
Thus, if $\vec x = (x_1, \ldots, x_m)$
and $\alpha = (\alpha_1, \ldots, \alpha_m)$
with each $\alpha_i \in \N_0^\Ucal$, and similarly for $\vec y=(y_1,\ldots,y_n)$ and
$\beta=(\beta_1,\ldots,\beta_n)$,
then
\begin{equation}
\nabla^{\alpha,\beta}_{\vec x,\vec y} g_{\vec x,\vec y}
  =
\nabla^{\alpha_1}_{x_1} \ldots \nabla^{\alpha_m}_{x_m}
\nabla^{\beta_1}_{y_1} \ldots \nabla^{\beta_n}_{y_n}  g_{x_1,\ldots,x_m,y_1,\ldots,y_n}.
\end{equation}

Let $\h_0 > 0$
be a parameter, which we set below.
We fix positive constants $p_\Phi\ge 4$ and
$p_\Ncal$
and assume that all test functions
vanish when $|\vec x|  +|\vec y| > p_\Ncal$.
For Theorem~\ref{thm:suscept}(i-ii), any choice of $p_\Ncal \ge 10$ is sufficient,
whereas for Theorem~\ref{thm:suscept}(iii) it is necessary to choose $p_\Ncal$ large
depending on $p$ \cite{BSTW-clp}.
The $\Phi = \Phi(\h_0)$ norm on such test functions is defined by
\begin{equation}
\|g\|_\Phi
  =
    \sup_{\vec x, \vec y} \h_0^{-(|\vec x| +|\vec y|)}
    \shift\shift \sup_{\alpha,\beta: |\alpha|_1+|\beta|_1 \le p_\Phi}
    |\nabla^{\alpha,\beta} g_{\vec x, \vec y}|,
\end{equation}
where $|\alpha|_1$ denotes the total order of the differential operator $\nabla^\alpha$.
Thus, for any test function $g$ and for sequences
$\vec x, \vec y$ with $|\vec x| +|\vec y| \leq p_\Ncal$ and
corresponding $\alpha, \beta$ with $|\alpha|_1 + |\beta|_1 \leq p_\Phi$,
\begin{equation}
\label{e:testfcnbd}
|\nabla^{\alpha,\beta} g_{\vec x,\vec y}| \leq \h_0^{|\vec x| + |\vec y|} \|g\|_\Phi.
\end{equation}

For any $F \in \Ncal$,
there exist \emph{unique} functions $F_{\vec y}$ of $(\phi, \bar\phi)$
that are anti-symmetric under permutations of $\vec y$, such that
\begin{equation}
F = \sum_{\vec y} \frac{1}{|\vec y|!} F_{\vec y}(\phi, \bar\phi) \psi^{\vec y}.
\end{equation}
Given a sequence $\vec{x}$ with $|\vec{x}| = m$, we define
\begin{equation}
F_{\vec x, \vec y} = \ddp{^m F_{\vec y}}{\phi_{x_1} \ldots \partial\phi_{x_m}}.
\end{equation}
We define a $\phi$-dependent pairing of elements of $\Ncal$ with test functions, by
\begin{equation}
\langle F, g \rangle_\phi
  =
\sum_{\vec x, \vec y} \frac{1}{|\vec x|! |\vec y|!} F_{\vec x,\vec y}(\phi, \bar\phi) g_{\vec x,\vec y}.
\end{equation}
Let $B(\Phi)$ denote the unit $\Phi$-ball
in the space of test functions. Then the
$T_\phi = T_\phi(\h_0)$ semi-norm on $\Ncal$
is defined by
\begin{equation}
\|F\|_{T_\phi} = \sup_{g\in B(\Phi)} |\langle F, g \rangle_\phi|.
\end{equation}

We need several properties of the $T_\phi$ semi-norm,
whose proofs can be found in \cite{BS-rg-norm}.
First, there is the important \emph{product property}
\cite[Proposition~\ref{norm-prop:prod}]{BS-rg-norm}
\begin{equation}
\label{e:prod}
\|F G\|_{T_\phi} \leq \|F\|_{T_\phi} \|G\|_{T_\phi}.
\end{equation}
An immediate consequence is that $\|e^{-F}\|_{T_\phi} \leq e^{\|F\|_{T_\phi}}$.
This is improved in \cite[Proposition~\ref{norm-prop:eK}]{BS-rg-norm},
which states that (recall that $F_\varnothing$ denotes the $0$-degree part of $F$)
\begin{equation}
\label{e:eK}
\|e^{-F}\|_{T_\phi} \leq e^{-2 {\rm Re} F_\varnothing(\phi) + \|F\|_{T_\phi}}.
\end{equation}

Each of the two choices $\varphi = \phi, \bar\phi$
can be viewed as a test function supported on sequences with
$|\vec x| = 1$ and $|\vec y| = 0$
and satisfying $\varphi_{\bar x} = \bar\varphi_x$.
In particular, $\|\phi\|_\Phi$ is defined as the norm of a test function.
We use \cite[Proposition~\ref{norm-prop:T0K}]{BS-rg-norm},
which states that if $F \in \Ncal$ is a polynomial in $\phi,\phib,\psi,\psib$ of
total degree $A \leq p_\Ncal$, then
\begin{equation}
\label{e:T0K}
\|F\|_{T_\phi} \leq \|F\|_{T_0} (1 + \|\phi\|_\Phi)^A.
\end{equation}

We write $x^\Box = \{y: |y-x|_\infty \le 2^d-1\}$,
where $|x|_\infty = \max\{|x_i| : 1 \le i \le d\}$
(this is the scale-0 version
of \cite[\eqref{IE-e:ssn}]{BS-rg-IE} for a single point).
The $\Phi_x \equiv \Phi(x^\square)$ norm of $\phi \in \C^\Lambda$ is defined by
\begin{equation}
\|\phi\|_{\Phi_x}
  =
\inf
\left\{
  \|\phi - f\|_\Phi : f \in \C^\Lambda \text{ such that } f_y = 0 \;\forall y \in x^\square
\right\}
.
\end{equation}
By taking the infimum in \eqref{e:T0K} over all possible
re-definitions of $\phi_y$ for $y \notin x^\square$, we get
\begin{equation}
\label{e:T0Kx}
\|F\|_{T_\phi}
  \leq
\|F\|_{T_0} (1 + \|\phi\|_{\Phi_x})^A
\end{equation}
when $F \in \Ncal(x^\square)$.

We need two choices of the parameter $\h_0$ (for both choices, $\h_0 \ge 1$):
either $\h_0 = \ell_0$, an $L$-dependent constant;
or $\h_0 = h_0 = k_0 \ggen_0^{-1/4}$, where $k_0$ is a small constant and
$\ggen_0$ is a constant which must be chosen small depending on $L$.
Some discussion of these constants occurs in the
proof of Proposition~\ref{prop:K0bd}.
In \cite{BS-rg-IE}, two \emph{regulators} are defined.
At scale $0$, these are given by
\begin{equation}
\lbeq{regdef}
G_0(x, \phi)
  = e^{\|\phi\|^2_{\Phi_x(\ell_0)}},
  \qquad
\tilde G_0(x, \phi)
  =
e^{\frac{1}{2} \|\phi\|^2_{\tilde\Phi_x(\ell_0)}}.
\end{equation}
The $\tilde \Phi_x$ norm in the definition of $\tilde G_0$,
is defined in \cite[\eqref{IE-e:Phitilnorm}]{BS-rg-IE};
it is a modification of the $\Phi_x$ norm that is invariant under shifts by
linear test functions.  Its specific properties do not play a direct role  in this paper.
Two regulator norms are defined for $F \in \Ncal(x^\square)$ by
\begin{equation}
\lbeq{reg0}
    \|F\|_{G_0} = \sup_{\phi\in\C^\Lambda} \frac{\|F\|_{T_\phi(\ell_0)}}{G_0(x,\phi)}
    , \quad
    \|F\|_{\tilde{G}^{\sf t}_0} = \sup_{\phi\in\C^\Lambda} \frac{\|F\|_{T_\phi(h_0)}}{\tilde{G}^{\sf t}_0(x,\phi)}
    ,
\end{equation}
where ${\sf t} \in (0, 1]$ is a constant power.


\subsection{Bounds on \texorpdfstring{$K_0$}{K0}}
\label{sec:K0bds}

The main estimate on $K^\pm_{0,x}$ is given by the following proposition.
Consistent with \cite[\eqref{IE-e:DV1-bis}]{BS-rg-IE}, we
fix a large constant $C_\DV$ and define
\begin{equation}
\label{e:DV0}
    \DV_0 = \DV_0(\ggen_0) = \{(g,\nu,z) \in \R^3 : C_{\DV}^{-1}\ggen_0 < g < C_{\DV}\ggen_0,
    \; |\nu|,|z| < C_{\DV}\ggen_0\}.
\end{equation}

\begin{prop}
\label{prop:K0bd}
Suppose that $V^\pm_0 \in \DV_0$, with $\ggen_0$ sufficiently small.
If $|\gamma_0| \leq  \ggen_0$, then
(with constants that may depend on $L$)
\begin{equation}
\lbeq{K0bds}
\|K^\pm_{0,x}\|_{G_0} = O(|\gamma_0|),
\quad
\|K^\pm_{0,x}\|_{\tilde G_0} = O(|\gamma_0|/g_0),
\end{equation}
where the bounds on $K^+$ and $K^-$ hold for $\gamma_0 \geq 0$
and $\gamma_0 < 0$, respectively.
\end{prop}

The form of the estimates \refeq{K0bds} can be anticipated from the definition of
$K_0^\pm$ in \refeq{Kpm}.  The upper bound arises from the small size of
$e^{-|\gamma_0|U_x^\pm}-1$.  For small fields, hence small $U_x^\pm$, this is of order $|\gamma_0|$,
as reflected by the $G_0$ norm estimate of \eqref{e:K0bds}.
For large fields, namely fields of size $|\phi| \approx \ggen_{0}^{-1/4}$, the difference
$e^{-|\gamma_0|U_x^\pm}-1$ is roughly of size $|\gamma_0|\,|\phi|^4 \approx |\gamma_0|/\ggen_0$.
This effect is measured by the $\tilde G_0$ norm.

Before proving the proposition, we
write \refeq{Kpm} for a singleton as
\begin{equation}
K^\pm_{0,x} = I^\pm_{0,x} J^\pm_x
  \label{e:KIJ},
\end{equation}
where, by the fundamental theorem of calculus,
\begin{align}
    I^\pm_{0,x} &= e^{-V^\pm_{0,x}} \\
    J^\pm_x
    &= e^{-|\gamma_0|U^\pm_x} - 1
    = - \int_0^{1} |\gamma_0| U^\pm_x e^{-t |\gamma_0| U^\pm_x} \; dt.
\label{e:J}
\end{align}
As in \eqref{e:Kpm}, the $+$ versions of \eqref{e:KIJ}--\eqref{e:J} hold
only for $\gamma_0 \geq 0$ and the $-$ versions only for $\gamma_0 < 0$.

Let $F \in \Ncal(x^\square)$ be a polynomial of degree at most $p_\Ncal$.
Then the stability estimates \cite[\eqref{IE-e:Iupper-a}--\eqref{IE-e:Iupper-b}]{BS-rg-IE}
imply that there exists $c_3 > 0$ and, for any $c_1 \geq 0$,
there exist positive constants $C, c_2$ such that
if $V_0^\pm \in \DV_0$ then
\begin{equation}
\label{e:Iupper}
\|I^\pm_{0,x} F\|_{T_\phi(\h_0)}
  \leq
C \|F\|_{T_0(\h_0)}
\begin{cases}
  e^{c_3 g_0 \left(1 + \|\phi\|^2_{\Phi_x(\ell_0)}\right)}
    & \h_0 = \ell_0 \\
  e^{-c_1 k_0^4 \|\phi\|^2_{\Phi_x(h_0)}} e^{c_2 k_0^4 \|\phi\|^2_{\tilde\Phi_x(\ell_0)}}
    & \h_0 = h_0.
\end{cases}
\end{equation}
This essentially reduces our task to estimating $J^\pm_x$.
The next lemma is an ingredient for this.

\begin{lemma}
\label{lem:FFnull-loc}
There is a universal constant $\tilde C$ such that
\begin{equation}
\label{e:FFnull}
\|U^\pm_x\|_{T_\phi(\h_0)}
  \leq
2 U^\pm_{\varnothing,x} + \tilde C \h_0^4 (1 + \|\phi\|^2_{\Phi_x(\h_0)}),
\end{equation}
where $U^\pm_\varnothing$ is the 0-degree part of $U^\pm$.
\end{lemma}

\begin{proof}
Let
\begin{equation}
M^+ = M^+_e = (\nabla^e \tau_x)^2,
\quad
M^- = M^-_e = 2 \tau_x \tau_{x+e},
\end{equation}
so that $U^\pm_x = \sum_{e\in\Ucal} M^\pm_e$.
It suffices to prove \eqref{e:FFnull} with $U^\pm_x$ replaced by $M^\pm$
(on both sides of the equation).
In addition, we can replace the $\Phi_x$ norm by the $\Phi$ norm;
the bound with the $\Phi_x$ norm then follows in the same way that \eqref{e:T0Kx} is a consequence of \eqref{e:T0K},
since $M^\pm \in \Ncal(x^\Box)$.

By definition of $\tau_x$,
\begin{equation}
M^\pm = M^\pm_{\varnothing} + R^\pm,
\end{equation}
where
\begin{alignat}{2}
&M^+_{\varnothing} = (\nabla^e |\phi_x|^2)^2,
  \quad
&&R^+ = 2 (\nabla^e |\phi_x|^2) \nabla^e (\psi_x\psib_x),
  \\
&M^-_\varnothing = 2 |\phi_x|^2 |\phi_{x+e}|^2,
  \quad
&&R^- = 2 (|\phi_x|^2 \psi_{x+e}\bar\psi_{x+e}
+ \psi_x\bar\psi_x |\phi_{x+e}|^2 + \psi_x\bar\psi_x\psi_{x+e}\bar\psi_{x+e}).
\end{alignat}
Thus, $\|M^\pm\|_{T_\phi} \leq \|M^\pm_{\varnothing}\|_{T_\phi} + \|R^\pm\|_{T_\phi}$.
A straightforward computation shows that
\begin{equation}
\label{e:Rpm-bound}
\|R^\pm\|_{T_\phi} = O(\h_0^4 (1 + \|\phi\|_\Phi)^2).
\end{equation}

By definition of the $T_\phi$ semi-norm,
\begin{equation}
\label{e:nabla-phi-sq-bd}
\|\nabla^e |\phi_x|^2\|_{T_\phi}
  \le
\nabla^e |\phi_x|^2 + 2 \h_0 (|\phi_x| + |\phi_{x+e}|) + 2 \h_0^2.
\end{equation}
Together with \eqref{e:Rpm-bound}, the product property,
and \eqref{e:testfcnbd}, this implies that
\begin{equation}
\|M^+\|_{T_\phi}
  \le
M^+_\varnothing
  + 2 |\nabla^e |\phi_x|^2| (2 \h_0 (|\phi_x| + |\phi_{x+e}|))
  + O(\h_0^4) (1 + \|\phi\|^2_\Phi).
\end{equation}
By the inequality
\begin{equation}
\label{e:young-ineq}
2|ab| \le |a|^2 + |b|^2
\end{equation}
and another application of \eqref{e:testfcnbd},
\begin{equation}
2 |\nabla^e |\phi_x|^2| (2 \h_0 (|\phi_x| + |\phi_{x+e}|))
  \le
M^+_\varnothing + O(\h_0^2 \|\phi\|^2_\Phi),
\end{equation}
and the bound on $M^+$ follows.

For the bound on $M^-$, we use the identity
\begin{equation}
\label{e:taunorm}
\|\tau_x\|_{T_\phi}
  =
(|\phi_x| + \h_0)^2 + \h_0^2
\end{equation}
from \cite[\eqref{norm-e:taunorm}]{BS-rg-norm}.
By the product property and \eqref{e:testfcnbd}, this implies that
\begin{equation}
\|M^-\|_{T_\phi}
  \le
2 |\phi_x|^2 |\phi_{x+e}|^2
  +
2 (|\phi_x| |\phi_{x+e}|) (2 \h_0 (|\phi_{x+e}| + |\phi_x|))
  +
O(\h_0^4) (1 + \|\phi\|^2_\Phi).
\end{equation}
Another application of \eqref{e:young-ineq} and \eqref{e:testfcnbd} gives
\begin{equation}
2 (|\phi_x| |\phi_{x+e}|) (2 \h_0 (|\phi_{x+e}| + |\phi_x|))
  \le
|\phi_x|^2 |\phi_{x+e}|^2 + O(\h_0^2 \|\phi\|^2_\Phi),
\end{equation}
and the proof is complete.
\end{proof}

An immediate consequence of Lemma~\ref{lem:FFnull-loc}, using \eqref{e:eK},
is that for any $s \ge 0$,
\begin{equation}
\label{e:Itilbd}
\|e^{-s U^\pm_x}\|_{T_\phi(\h_0)} \leq e^{\tilde C s \h_0^4 (1 + \|\phi\|^2_{\Phi_x(\h_0)})}.
\end{equation}

\begin{proof}[Proof of Proposition~\ref{prop:K0bd}]
According to the definition of the
regulator norms in \refeq{regdef}--\refeq{reg0},
it suffices to prove that, under the hypothesis on $\gamma_0$,
\begin{equation}
\label{e:K0bd}
  \|K^\pm_{0,x}\|_{T_\phi(\h_0)} = O(|\gamma_0| \h_0^4)
  \begin{cases}
  e^{\|\phi\|_{\Phi_x}^2} & (\h_0=\ell_0)
  \\
  e^{\frac{{\sf t}}{2} \|\phi\|_{\tilde\Phi}} & (\h_0=h_0).
  \end{cases}
\end{equation}
For $t \in [0,1]$, let $\tilde I^\pm_x(t) = e^{-t |\gamma_0| U^\pm_x}$.
By \eqref{e:KIJ}, \eqref{e:J}, and the product property,
\begin{align}
\label{e:K0x-est}
    \|K^\pm_{0,x}\|_{T_\phi(\h_0)}
    & \le |\gamma_0| \|I^\pm_{0,x} U^\pm_x\|_{T_\phi(\h_0)}
    \sup_{t\in [0, 1]} \|\tilde I^\pm_{x}(t)\|_{T_\phi(\h_0)}.
\end{align}
By \refeq{Iupper} and Lemma~\ref{lem:FFnull-loc},
there exists $c_3 > 0$, and, for any $c_1 \geq 0$ there exists $c_2 > 0$, such that
\begin{equation}
\label{e:Iupper-bis}
\|I^\pm_{0,x} U^\pm_x\|_{T_\phi(\h_0)}
  \leq
O(\h_0^4)
\begin{cases}
  e^{c_3 g_0  \|\phi\|^2_{\Phi_x(\ell_0)}}
    & \h_0 = \ell_0 \\
  e^{-c_1 k_0^4 \|\phi\|^2_{\Phi_x(h_0)}} e^{c_2 k_0^4 \|\phi\|^2_{\tilde\Phi_x(\ell_0)}}
    & \h_0 = h_0.
\end{cases}
\end{equation}
The constant in $O(|\gamma_0| \h_0^4)$ may depend on $c_1$,
but this is unimportant.
Also, by \eqref{e:Itilbd},
\begin{equation}
\sup_{t\in[0,1]} \|\tilde I_{x}^\pm(t) \|_{T_\phi(\h_0)}
  \le
e^{\tilde C |\gamma_0| \h_0^4 (1+\|\phi\|^2_{\Phi_x(\h_0)})}.
\end{equation}

Thus, for $\h_0=\ell_0$,
the total exponent in our estimate for the right-hand side of \refeq{K0x-est}
is
\begin{equation}
    \tilde C |\gamma_0| \ell_0^4
       +(c_3 g_0 + \tilde C |\gamma_0| \ell_0^4) \|\phi\|^2_{\Phi_x(\ell_0)}
     .
\end{equation}
This gives the $\h_0=\ell_0$ version of \refeq{K0bd} provided that
$g_0$ is small and $|\gamma_0|$ is small depending on $L$.

For $\h_0=h_0$, the total exponent in our estimate for the right-hand side of \refeq{K0x-est}
is
\begin{equation}
    \tilde C |\gamma_0| k_0^4 \ggen_0^{-1}
        + (\tilde C |\gamma_0| k_0^4 \ggen_0^{-1} - c_1 k_0^4) \|\phi\|^2_{\Phi_x(h_0)}
        + c_2 k_0^4 \|\phi\|^2_{\tilde\Phi_x(\ell_0)}.
\end{equation}
This gives the $\h_0=h_0$ version of \refeq{K0bd} provided that
$|\gamma_0| \le \ggen_0$, $c_1\ge \tilde C$, and $c_2 k_0^4 \le {\sf t}/2$.

All the provisos are satisfied
if we choose
$c_1 \ge \tilde C$,
$k_0$ small depending on $c_1$
and $\ggen_0$ small.
\end{proof}

\begin{rk}
By a small modification to the proof of Proposition~\ref{prop:K0bd},
it can be shown that if $M_x \in \Ncal(x^\square)$ is a monomial of
degree $r \le p_\Ncal -4$ (so that $M_xU_x^\pm$ has degree at most $p_\Ncal$), then
\begin{equation}
\label{e:K0bd-gen}
\|M_x K^\pm_{0,x}\|_{\Gcal_0} = O(|\gamma_0| \h_0^{4+r}).
\end{equation}
\end{rk}

\subsection{Unified bound on \texorpdfstring{$K_0$}{K0}}
\label{sec:KWcal}

The results of \cite{BS-rg-step,BBS-rg-flow} are formulated in a sequence of spaces $\Wcal_j$ that
enable the combination of small-field and large-field estimates into a single norm estimate.
In this section, we recast the result of Proposition~\ref{prop:K0bd} to see that $K_0^\pm$
fits into this formulation.

We restrict attention in this section to the $\Wcal_0$ norm,
whose definition is recalled below.
This requires several preliminaries.
Let $\Pcal_0 = \Pcal_0(\Lambda)$ denote the collection of subsets of vertices in $\Lambda$.
We refer to the elements of $\Pcal_0$ as \emph{polymers}.
We call a nonempty polymer $X\in \Pcal_0$ \emph{connected}
if for any $x, x' \in X$, there is a sequence
$x = x_0, \ldots, x_n = x' \in X$ such that
$|x_{i+1} - x_i|_\infty = 1$ for $i = 0, \ldots, n - 1$.
Let $\Ccal_0$ denote the set of connected polymers.
The \emph{small set neighbourhood} $X^\Box$ of $X\in\Pcal_0$ is defined by
\begin{equation}
    X^\Box =
    \{y \in \Lambda : \exists x \in \Lambda \; \text{such that}\; |y-x|_\infty \le 2^d\}.
\end{equation}
We extend the definitions of the regulators $\Gcal_0 = G_0, \tilde G_0^{\sf t}$,
defined in \refeq{regdef}, by setting
\begin{equation} \label{e:Gcalprod}
\Gcal_0(X, \phi) = \prod_{x\in X} \Gcal_0(x, \phi),
\end{equation}
and extend the definitions \refeq{reg0} to define norms, for $F \in \Ncal(X^\Box)$, by
\begin{equation}
\lbeq{reg0X}
    \|F\|_{G_0} = \sup_{\phi\in\C^\Lambda} \frac{\|F\|_{T_\phi(\ell_0)}}{G_0(X,\phi)}
    , \quad
    \|F\|_{\tilde{G}^{\sf t}_0} = \sup_{\phi\in\C^\Lambda} \frac{\|F\|_{T_\phi(h_0)}}{\tilde{G}^{\sf t}_0(X,\phi)}
    .
\end{equation}
It follows from the product property of the $T_\phi$ norm that these norms obey the product property
\begin{equation}
    \|F_1F_2\|_{\Gcal_0} \le   \|F_1\|_{\Gcal_0} \|F_2\|_{\Gcal_0}
    \quad \text{for $F_i\in \Ncal(X_i^\Box)$ with $X_1 \cap X_2=\varnothing$.}
\end{equation}

Given a map $K: \Pcal_0 \to \Ncal$ with the property that $K(X) \in \Ncal(X^\Box)$
for all $X \in \Pcal_0$,
we define the $\Fcal_0(\Gcal)$ norms (for $\Gcal = G, \tilde G$) by
\begin{align}
\|K\|_{\Fcal_0(G)}        &= \sup_{X\in\Ccal_0} \ggen_0^{-f_0(a, X)} \|K(X)\|_{G_0} \\
\|K\|_{\Fcal_0(\tilde G)} &= \sup_{X\in\Ccal_0}
\ggen_0^{-f_0(a, X)} \|K(X)\|_{\tilde G_0^{\sf t}},
\end{align}
with
\begin{equation}
    \label{e:f0def}
    f_0 (\amain, X)
    =
    \amain (|X|-2^d)_+
    =
    \begin{cases}
    a (|X| - 2^d)
    & \text{if } |X| > 2^d   \\
    0
    & \text{otherwise}.
    \end{cases}
\end{equation}
Here $a$ is a small constant;  its value is discussed below \cite[\eqref{step-e:T0dom}]{BS-rg-step}.
The $\Wcal_0$ norm is then defined by
\begin{align}
\label{e:9Kcalnorm}
\|K\|_{\Wcal_0}
  &=
  \max
  \Big\{
  \|K \|_{\Fcal_0(G)},\,
  \ggen_0^{9/4}
  \|K \|_{\Fcal_0(\tilde{G})}
  \Big\}.
\end{align}
Since this definition depends on $\ggen_0$ and the
volume $\Lambda$, we sometimes write $\Wcal_0 = \Wcal_0(\ggen_0, \Lambda)$.
The following proposition uses Proposition~\ref{prop:K0bd} to obtain a bound on the $\Wcal_0$ norm
of the map $K_0^\pm : \Pcal_0 \to \Ncal$ defined by
\begin{equation}
    K_0^\pm(X) = \prod_{x \in X} K_{0,x}^\pm \qquad (X \in \Pcal_0)
    .
\end{equation}

\begin{prop}
\label{prop:KWcal}
If $V_0^\pm \in \DV_0$ with $\ggen_0$ sufficiently small
(depending on $L$), and if $|\gamma_0| \le O(\ggen_0^{1+a'})$
for some $a' >a$,
then $\|K_0^\pm\|_{\Wcal_0} \le O(|\gamma_0|)$,
where all constants may depend on $L$.
\end{prop}

\begin{proof}
Let $X$ be a connected polymer in $\Pcal_0$.
By the product property and Proposition~\ref{prop:K0bd},
\begin{align}
\lbeq{K0prod}
    \|K_0^\pm(X)\|_{\Gcal_0} \le (c|\gamma_0|\h_0^4)^{|X|}
    &=
    (c|\gamma_0|\h_0^4)^{|X|\wedge 2^d} (c|\gamma_0|\h_0^4)^{(|X|-2^d)_+}.
\end{align}
For $\Gcal_0=G_0$, we use $\h_0=\ell_0$,
$(c|\gamma_0|\h_0^4)^{|X|\wedge 2^d}\le O(|\gamma_0|)$, and
\begin{equation}
    (c|\gamma_0|\h_0^4)^{(|X|-2^d)_+} \le (c' \ggen_0)^{(1+a')(|X|-2^d)_+} \le \ggen_0^{f_0(a,X)}.
\end{equation}
For $\Gcal_0=\tilde G_0$, we use $\h_0=h_0 = O(\ggen_0^{-1/4})$ and, since $a'>a$,
\begin{equation}
    (c|\gamma_0|\h_0^4)^{(|X|-2^d)_+} \le (c' \ggen_0)^{a'(|X|-2^d)_+} \le \ggen_0^{f_0(a,X)}.
\end{equation}
Since $|\gamma_0| \le \ggen_0$, it follows from \refeq{K0prod} that
\begin{equation}
    \ggen_0^{9/4}   \|K_0^\pm \|_{\Fcal_0(\tilde{G})}
    \le
    \ggen_0^{9/4}O(|\gamma_0| \ggen_0^{-1})
    \le |\gamma_0|,
\end{equation}
and the proof is complete.
\end{proof}

The above discussion is based on norms in the setting of the torus $\Lambda$.
As in \cite{BS-rg-step}, a version on the infinite lattice $\Zd$ is also required.
This can be done in exactly the same manner,
by defining
the polymers $\Pcal_0 = \Pcal_0(\Zd)$
to be the collection
of subsets of $\Zd$, with $K_0^\pm(X)$ defined for subsets of $\Zd$ by
$\prod_{x \in X} K_{0,x}^\pm$.
The $\Wcal_0 = \Wcal_0(\ggen_0, \Zd)$ norm (in infinite volume)
can be defined analogously to \eqref{e:9Kcalnorm}.
The hypotheses and conclusion of Proposition~\ref{prop:KWcal} remain the same
in the setting of $\Zd$.

\subsection{Smoothness of \texorpdfstring{$K_0$}{K0}}
\label{sec:Ksmooth}

Let $\Ccal_0(\Z^d) \subset \Pcal_0(\Z^d)$ be the set of connected polymers.
By definition, a connected polymer is nonempty.
Given $\ggen_0>0$, let
$\Wcal^*_0(\ggen_0, \Zd)$ denote the space of maps
$F :\Ccal_0(\Zd) \to \Ncal$,
with $F(X) \in \Ncal(X^\Box)$ and $\|F\|_{\Wcal_0(\ggen_0, \Zd)} < \infty$.
Addition in this space is defined by $(F_1+F_2)(X)=F_1(X)+F_2(X)$.
We extend any $F :\Ccal_0(\Zd) \to \Ncal$ to $F :\Pcal_0(\Zd) \to \Ncal$
by taking $F(X) = \prod_{Y} F(Y)$ where the product is over the connected components $Y$ of $X$.

Given any map $F : D \to \Wcal^*_0(\ggen_0, \Zd)$ for $D \subset \R$ an open interval,
write $F_X, F^\phi_X : D \to \Ncal(X^\square)$ for the
maps defined by partial evaluation of $F$ at $X$ and at
$(X, \phi)$, respectively. We say $F^\phi_X$ is $C^k$
if all of its coefficients in the decomposition \eqref{e:FinNcal}
are $C^k$ as functions $D \to \R$.

\begin{lemma}
\label{lem:smoothness}
Let $D \subset \R$ be open and $F : D \to \Wcal^*_0(\ggen_0, \Zd)$ be a map.
Suppose that $F^\phi_X$ is $C^2$ for all $X \in \Ccal_0$
and $\phi \in \C^\Lambda$, and define 
$F^{(i)} : D \to \Wcal^*_0(\ggen_0, \Zd)$ by $(F^{(i)}(t))^\phi_X = (F^\phi_X)^{(i)}(t)$ for $i = 1, 2$,
where the right-hand side denotes the (component-wise) $i^{\rm th}$
derivative of $F^\phi_X$.
If $\|F^{(i)}(t)\|_{\Wcal_0} < \infty$ for $i = 1, 2$ and $t \in D$, then 
$F^{(1)}$ is the derivative of $F$.
\end{lemma}

\begin{proof}
For $t, t + s \in D$, define $R(t, s) \in \Wcal_0$ by
\begin{equation}
R^\phi_X(t, s) = F^\phi_X(t + s) - F^\phi_X(t) - s (F^\phi_X)'(t).
\end{equation}
By Taylor's theorem, for any $\phi$ and $X$,
\begin{equation}
R^\phi_X(t, s) = s^2 \int_0^1 (F^\phi_X)''(t + u s) (1 - u) \; du,
\end{equation}
where the integral is taken component-wise.
It follows
that
\begin{equation}
\|R(t, s)\|_{\Wcal_0}
  \le |s|^2 \sup_{u\in[0,1]} \|F''(t+us)\|_{\Wcal_0}
  \le O(|s|^2),
\end{equation}
so $F$ is differentiable 
and its derivative satisfies $(F')^\phi_X = (F^\phi_X)'$.
Continuity of $F'$ follows similarly, since, by the
fundamental theorem of calculus,
\begin{equation}
\|F'(t+s) - F'(t)\|_{\Wcal_0}
  \le
|s| \sup_{u\in[t,t+s]} \|F''(u)\|_{\Wcal_0}
  \le
O(|s|),
\end{equation}
which suffices.
\end{proof}

Consider the map
\begin{equation}
(g_0, \gamma_0, \nu_0, z_0) \mapsto K_0 \in \Wcal^*_0(\ggen_0, \Zd)
\end{equation}
defined by
\begin{equation}
\label{e:K0def}
K_0(g_0, \gamma_0, \nu_0, z_0) =
\begin{cases}
K^+_0(g_0, \gamma_0, \nu_0, z_0)
  & (\gamma_0 \geq 0) \\
K^-_0(g_0, \gamma_0, \nu_0, z_0)
  & (\gamma_0 < 0),
\end{cases}
\end{equation}
for $(g_0, \gamma_0, \nu_0, z_0)$ satisfying the hypotheses
of Proposition~\ref{prop:KWcal}.
The map $K_0$ is in fact analytic away from $\gamma_0 = 0$.
However, we only prove the following, which is what we need later.

\begin{prop}
\label{prop:Ksmooth}
Suppose that $V_0^\pm \in \DV_0$, with $\ggen_0$ sufficiently small
(depending on $L$) and $|\gamma_0| \le O(\ggen_0^{1+a'})$
for some $a' >a$.
The map $K_0(g_0, \gamma_0, \nu_0, z_0)$ is jointly continuous
in its four variables, is
$C^1$ in $(g_0, \nu_0, z_0)$,
and (when $\gamma_0 \ne 0$) is $C^1$ in $(g_0, \gamma_0, \nu_0, z_0)$,
with partial derivatives with respect to $t = g_0$, $\nu_0$, and $z_0$ satisfying
\begin{equation}
\label{e:ddpK}
\|\partial K_0 / \partial t\|_{\Wcal_0} = O(|\gamma_0| \h_0^8).
\end{equation}
Moreover, $K_0$
is left- and right-differentiable in $\gamma_0$ at $\gamma_0 = 0$.
\end{prop}

\begin{proof}
Let $t$ denote any one of the coupling constants $g_0, \gamma_0, \nu_0$ or $z_0$.
We drop the subscript $0$, and let $K(t)$ denote $K_0$ viewed as a function of $t$,
with the remaining coupling constants fixed. Then $K^\phi_X$ is smooth for any $\phi, X$.
If $t$ is $g_0, \nu_0$ or $z_0$, then
\begin{align}
(K^\phi_x)'  &= -M_x(\phi) K^\phi_x, \quad
(K^\phi_x)'' = M_x^2(\phi) K^\phi_x,
\end{align}
where $M_x$ is $\tau_x^2, \tau_x$ or $\tau_{\Delta,x}$, respectively.
The maximal degree of $M_x$ is $4$, so
\eqref{e:K0bd-gen} implies that
\begin{equation}
\label{e:Kprime-bd1}
\|K'_x\|_{\Gcal_0} \le O(|\gamma_0| \h_0^{8}),
  \quad
\|K''_x\|_{\Gcal_0} \le O(|\gamma_0| \h_0^{12}).
\end{equation}

For $t$ denoting $\gamma_0$,
we restrict attention to $\gamma_0 > 0$, and write $U = U^+$
and $V_0 = V^+_0$ (the case $\gamma_0 < 0$ is similar). Then
\begin{equation}
\label{e:dKdgamma0}
(K^\phi_x)'  = -U_x(\phi) e^{-V_x(\phi) - \gamma_0 U_x(\phi)}, \quad
(K^\phi_x)'' = U_x^2(\phi) e^{-V_x(\phi) - \gamma_0 U_x(\phi)},
\end{equation}
and \eqref{e:Iupper} and \eqref{e:Itilbd} imply that
\begin{equation}
\label{e:Kprime-bd2}
\|K'_x\|_{\Gcal_0} \le O(\h_0^4),
  \quad
\|K''_x\|_{\Gcal_0} \le O(\h_0^8).
\end{equation}

By definition, $K_X = \prod_{x \in X} K_x$, so, for derivatives with respect to any one
of the four variables (with $\gamma_0 \neq 0$ when differentiating with respect to $\gamma_0$),
\begin{equation}
\label{e:KXprime}
(K^\phi_X)'  = \sum_{x \in X} (K^\phi_x)' K^\phi_{X \setminus x}, \quad
(K^\phi_X)'' = \sum_{x \in X} ((K^\phi_x)'' K^\phi_{X \setminus x} + (K^\phi_x)' (K^\phi_{X \setminus x})').
\end{equation}
Thus, by the product property, \eqref{e:Kprime-bd1}, and Proposition~\ref{prop:K0bd},
\begin{equation}
\|K'_X\|_{\Gcal_0}
  \le
O(|X|) |\gamma_0| \h_0^8 (|\gamma_0| \h_0^4)^{|X|-1}.
\end{equation}
when differentiating with respect to $g_0$, $\nu_0$, or $z_0$.
The bound \eqref{e:ddpK} then follows from the hypothesis on $\gamma_0$.
Similarly, using \eqref{e:Kprime-bd2},
\begin{equation}
\|K'_X\|_{\Gcal_0}
  \le
O(|X|) \h_0^4 (|\gamma_0| \h_0^4)^{|X|-1}
\end{equation}
when differentiating with respect to $\gamma_0$ away from $\gamma_0 = 0$.
In both cases, we have
\begin{equation}
\|K''_X\|_{\Gcal_0}
  \le
O(|X|^2) \h_0^8 (|\gamma_0| \h_0^4)^{(|X|-2) \wedge 0}.
\end{equation}
Thus, by Lemma~\ref{lem:smoothness}, $K$ is $C^1$ in any of its variables.
Therefore, $K$ is $C^1$ in $(g_0, \nu_0, z_0)$ on the whole domain and in all the variables when $\gamma_0 \ne 0$.

To show right-continuity in $\gamma_0$ at $\gamma_0 = 0$,
fix $(g_0, \nu_0, z_0)$ and define $F \in \Wcal^*_0$ by
\begin{equation}
F(X) =
\begin{cases}
  -U_x e^{-V_{0,x}}
    & X = \{ x \} \\
  0 & |X| > 1,
\end{cases}
\end{equation}
where $U_x, V_{0,x}$ are defined above.
Let $K'(\gamma_0)$ denote the $\gamma_0$ derivative of $K$ evaluated at $\gamma_0 > 0$.  Then
\eqref{e:dKdgamma0} and \eqref{e:KXprime} imply that
\begin{equation}
F(X) - K'_X(\gamma_0)
  =
\begin{cases}
  U_x K_x(\gamma_0)
    & X = \{ x \} \\
  \sum_{x \in X} K'_x(\gamma_0) K_{X \setminus x}(\gamma_0)
    & |X| > 1.
\end{cases}
\end{equation}
Thus, by \eqref{e:K0bd-gen}, \eqref{e:Kprime-bd2}, and Proposition~\ref{prop:K0bd},
\begin{equation}
\|F(X) - K'_X(\gamma_0)\|_{\Gcal_0}
  \le
\begin{cases}
  O(\gamma_0 \h_0^8)
    & X = \{ x \} \\
  O(|X|) \h_0^4 (\gamma_0 \h_0^4)^{|X|-1}
    & |X| > 1.
\end{cases}
\end{equation}
It follows that
\begin{equation}
\lim_{\gamma_0\downarrow 0} \|F - K'(\gamma_0)\|_{\Wcal_0} = 0,
\end{equation}
i.e., $F$ is the right-derivative of $K$ in $\gamma_0$ at $\gamma_0 = 0$.
Left-continuity is handled similarly.
\end{proof}


\section{Existence of critical parameters}
\label{sec:nucident}

In Sections~\ref{sec:rg-coords}--\ref{sec:rg-map}, we recall some facts about the renormalisation group map
defined in \cite{BS-rg-step}. In Section~\ref{sec:flow}, we discuss the existence and properties of the
finite-volume renormalisation group flow (a consequence of the main result of \cite{BBS-rg-flow}),
which is crucial to proving Theorem~\ref{thm:suscept}.
Using the results of Section~\ref{sec:flow}, we identify critical initial conditions for
iteration of the
renormalisation group in Section~\ref{sec:nu0z0c}.
In Section~\ref{sec:nuztilde}, we identify the critical point and
discuss an important change of parameters.
Then in Section~\ref{sec:conclusion} we
obtain the asymptotic behaviour of the two-point function, susceptibility, and correlation
length of order $p$, and thereby prove  Theorem~\ref{thm:suscept}.  Finally,
Section~\ref{sec:IFT} contains a version of the implicit function theorem that we apply
in Sections~\ref{sec:nu0z0c}--\ref{sec:nuztilde}.

\subsection{Renormalisation group coordinates}
\label{sec:rg-coords}

As discussed in Section~\ref{sec:prog},
the evolution of $Z_j$ defined in \eqref{e:Zjdef}
is tracked via coordinates $(I_j, K_j)$.
In order to discuss these, we make the following definitions.
We partition $\Lambda$ into $L^{N-j}$ disjoint \emph{scale-$j$ blocks}
of side $L^j$.
We let $\Pcal_j$ denote the set of \emph{scale-$j$ polymers},
which are unions of elements of $\Bcal_j$.  Given $X \in \Pcal_j$,
we denote the collection of scale-$j$ blocks in  $X$
by $\Bcal_j(X)$.
Scale-$0$ blocks
are simply elements of $\Lambda$, and scale-$0$ polymers are subsets of $\Lambda$,
as in Section~\ref{sec:KWcal}.
Also, as in the scale-$0$ case,
there is a version of blocks and polymers also on $\Zd$ rather than $\Lambda$.

Given  a polynomial $V_j$ of the form
\begin{equation}
V_{j;x} = g_j \tau_x^2 + \nu_j \tau_x + z_j \tau_{\Delta,x},
\end{equation}
the \emph{interaction} $I_j(X)$ is defined for $X \in \Pcal_j(\Lambda)$ by
\begin{equation}
I_j(X) = e^{-\sum_{x\in X} V_{j;x}} \prod_{B \in \Bcal_j(X)} (1 + W_j(B)),
\end{equation}
where $W_j(B)$ is an explicit polynomial that is quadratic in $V_j$ and is defined in
\cite[\eqref{pt-e:WLTF}]{BBS-rg-pt}.
In \cite[Definition~\ref{step-def:Kspace}]{BS-rg-step},
a space $\Kcal_j = \Kcal_j(\Lambda)$ of maps $\Pcal_j \to \Ncal$ required to satisfy
several properties is defined.
The coordinate $K_j$ is constructed in \cite{BS-rg-step} as an element of $\Kcal_j$.
The renormalisation group is used to construct a sequence $(V_j, K_j)$
from which $Z_j$ can be recovered via the \emph{circle product}
\begin{equation}
\label{e:ZjIjKj}
Z_j = (I_j \circ K_j)(\Lambda) = \sum_{X\in\Pcal_j(\Lambda)} I_j(\Lambda\setminus X) K_j(X).
\end{equation}

\subsection{Renormalisation group map}
\label{sec:rg-map}

We restrict the discussion in this section to a finite volume $\Lambda = \Lambda_N$
with $N > 1$.

For fixed $(\mgen^2, \ggen_0) \in [0, \delta) \times (0, \delta)$,
we define a sequence $\ggen_j = \ggen_j(\mgen^2, \ggen_0)$ as in
\cite[\eqref{log-e:ggendef}]{BBS-saw4-log};
in particular, $\ggen_0(\mgen^2, \ggen_0) = \ggen_0$.
In \cite[Section~\ref{step-sec:Knorms}]{BS-rg-step},
a sequence of norms $\|\cdot\|_{\Wcal_j} = \|\cdot\|_{\Wcal_j(\mgen^2, \ggen_j, \Lambda)}$
parametrised by $(\mgen^2, \ggen_j)$ is defined on maps $\Pcal_j \to \Ncal$.
We let $\Wcal_j$ denote the subspace of $\Kcal_j$
consisting of all elements having finite $\Wcal_j$ norm.
Note that $\Wcal_0 = \Kcal_0 \cap \Wcal_0^*$,
where $\Wcal_0^*$ is defined in Section~\ref{sec:Ksmooth}.

In \cite[\eqref{log-e:mass-scale}--\eqref{log-e:chidef}]{BBS-saw4-log},
a function $\chicCov_j = \chicCov_j(m^2)$ (denoted $\chi_j$ in \cite{BBS-saw4-log})
is defined in such a way that $\chicCov_j$ decays exponentially
when $j$ is sufficiently large depending on $m$.
We write $\chicCovgen_j = \chicCov_j(\mgen^2)$.
Given a constant $\alpha > 0$,
we define the (finite-volume)
renormalisation group domains $\domRG_j \subset \R^3 \oplus \Wcal_j$ by
\begin{equation}
\label{e:domRG}
\domRG_j(\mgen^2, \ggen_j, \Lambda)
  =
\DV_j
  \times
B_{\Wcal_j(\mgen^2, \ggen_j, \Lambda)}(\alpha \chicCovgen_j \ggen_j^3),
\end{equation}
\begin{equation}
\DV_j = \DV_j(\ggen_j)
=
\{ (g, \nu, z) : C^{-1}_\DV \ggen_j < g < C_\DV \ggen_j, |z|, L^{2j} |\nu| < C_\DV \ggen_j \}.
\end{equation}
This definition of $\DV_j$ is consistent with \eqref{e:DV0} when $j = 0$.
We let $\Igen_j(\mgen^2)$ be the neighbourhood of $\mgen^2$ defined by
\begin{equation}
\lbeq{Itilint}
    \Igen_j = \Igen_j(\mgen^2) =
    \begin{cases}
    [\frac 12 \mgen^2, 2 \mgen^2] \cap \Iint_j & (\mgen^2 \neq 0)
    \\
    [0,L^{-2(j-1)}] \cap \Iint_j & (\mgen^2 =0)
    \end{cases},
\end{equation}
where $\Iint_j = [0, \delta]$ if $j < N$ and $\Iint_N = [\delta L^{-2 (N - 1)}, \delta]$.
The main result of \cite{BS-rg-step} is the construction of the
renormalisation group map on the domains $\domRG_j$.
Although \cite{BS-rg-step} constructs finite- and infinite-volume
versions of this map, we only discuss the finite-volume map here.

For $m^2 \in \Igen_j(\mgen^2)$, the finite-volume renormalisation group
map at scale $j = 1, \ldots, N - 1$ is a map
$\domRG_j(\mgen^2, \ggen_j, \Lambda) \to \R^3 \oplus \Wcal_{j+1}(\mgen^2, \ggen_{j+1}, \Lambda)$,
which we denote
\begin{equation}
\label{e:RGmap}
(V_j, K_j) \mapsto (V_{j+1}, K_{j+1}).
\end{equation}
The first component of this map takes the form
\begin{equation}
\label{e:Vflow}
V_{j+1} = V_{\pt,j+1}(V_j) + R_{j+1}(V_j, K_j),
\end{equation}
where the map $V_{\pt,j+1}$ defined in \cite{BBS-rg-pt}
captures the second-order evolution of $V_j$, and $R_{j+1}$
is a third-order contribution.
The main properties of the map \eqref{e:RGmap} are listed in
\cite[Section~\ref{log-sec:rgiv}]{BBS-saw4-log}.
Importantly, the renormalisation group map preserves the
circle product in the sense that
\begin{equation}
(I_{j+1} \circ K_{j+1})(\Lambda) = \Ex_{C_{j+1}}\theta (I_j \circ K_j)(\Lambda).
\end{equation}
Since
$\Pcal_N(\Lambda)=\{\varnothing,\Lambda_N\}$,
this means that, if $(V_0, K_0) = (V^\pm_0, K^\pm_0)$ and if the renormalisation
group map can be iterated $N$ times with this choice of initial condition, then
\begin{equation}
\label{e:ZVKN}
Z_N = I_N(\Lambda) + K_N(\Lambda) = e^{-\sum_{x\in\Lambda} V_{N;x}} (1 + W_N(\Lambda)) + K_N(\Lambda).
\end{equation}

\subsection{Renormalisation group flow}
\label{sec:flow}

The following theorem is an extension of \cite[Proposition~\ref{log-prop:flow-flow}]{BBS-saw4-log}
to non-trivial $K_0$. Such an extension is possible,
with only minor modifications to the proof of the $K_0 = \1_\varnothing$ case,
due to the generality allowed by the main result of \cite{BBS-rg-flow}.

The theorem provides, for any $N \ge 1$ and for initial error coordinate $K_0$
in a specified domain, a choice of initial condition $(\nu_0^c,z_0^c)$
for which there exists
a finite-volume renormalisation group flow $(V_j, K_j) \in \domRG_j$ for $0 \le j \le N$.
In order to ensure a degree of consistency amongst the sequences $(V_j, K_j)$, which depend on
the volume $\Lambda_N$, a notion of consistency must be imposed upon the collection of initial
error coordinates $K_{0,\Lambda} \in \Kcal_0(\Lambda)$ for varying $\Lambda$.
Specifically, the family $K_{0,\Lambda}$ is required to satisfy the property $(\Zd)$ of
\cite[Definition~\ref{step-defn:KZd}]{BS-rg-step}.
We refer to any such family as a $\Lambda$-family.
As discussed in \cite[Definition~\ref{step-defn:KZd}]{BS-rg-step},
any $\Lambda$-family
induces an infinite-volume error coordinate $K_{0,\Zd} \in \Kcal_0(\Zd)$ in a natural way.

\begin{theorem}
\label{thm:flow-flow}
Let $d = 4$.
There exists a constant $a_* > 0$ and continuous functions $\nu_0^c, z_0^c$
of $(m^2, g_0, K_0)$, defined for $(m^2, g_0) \in [0, \delta]^2$
(for some $\delta > 0$ sufficiently small) and for any $K_0 \in \Wcal_0(m^2, g_0, \Zd)$
with $\|K_0\|_{\Wcal_0(m^2, g_0, \Zd)} \leq a_* g_0^3$, such that
the following holds for $g_0 > 0$:
if $K_{0,\Lambda} \in \Kcal_0(\Lambda)$ is a $\Lambda$-family
that induces the infinite-volume coordinate $K_0$, and if
\begin{equation}
\label{e:flow-flow-ic}
V_0 = V_0^c(m^2, g_0, K_0) = (g_0, \nu_0^c(m^2,g_0,K_0), z_0^c(m^2,g_0,K_0)),
\end{equation}
then for any $N \in \N$ and $m^2 \in [\delta L^{-2 (N - 1)}, \delta]$,
there exists a sequence $(V_j, K_j) \in \domRG_j(m^2, g_0, \Lambda)$
such that
\begin{equation}
  \label{e:VjKjDj}
  (V_{j+1},K_{j+1}) = (V_{j+1}(V_j, K_j), K_{j+1}(V_j, K_j)) \text{ for all } j < N
\end{equation}
and \eqref{e:ZjIjKj} is satisfied.
Moreover, $\nu_0^c,z_0^c$ are continuously differentiable in
$g_0 \in (0, \delta)$ and $K_0 \in B_{\Wcal_0(m^2, g_0, \Lambda)}(a_* g_0^3)$, and
\begin{align}
&\nu_0^c(m^2,0,0) = z_0^c(m^2,0,0) = 0,
\quad
\ddp{\nu_0^c}{g_0} = O(1),
\quad
\label{e:z0est}
\ddp{z_0^c}{g_0} = O(1),
\end{align}
where the estimates above hold uniformly in $m^2 \in [0, \delta]$.
\end{theorem}

\begin{proof}
The proof results from small modifications to the proofs of
\cite[Proposition~\ref{log-prop:flow-flow}]{BBS-saw4-log} and then to
\cite[Proposition~\ref{log-prop:KjNbd}]{BBS-saw4-log},
where (in both cases) we relax the requirement that $K_0 = \1_\varnothing$,
which was chosen in \cite{BBS-saw4-log} due to the fact that
$K_0 = \1_\varnothing$ when $\gamma=0$.
The more general condition that $K_0 \in B_{\Wcal_0(m^2, g_0, \Lambda)}(a_* g_0^3)$
comes from the hypothesis of \cite[Theorem~\ref{flow-thm:flow}]{BBS-rg-flow}
when $(m^2, g_0) = (\mgen^2, \ggen_0)$.
By \cite[Remark~\ref{flow-rk:Nrad}]{BBS-rg-flow}, no major changes to the proof
result from this choice of $K_0$.
The following paragraph outlines
in more detail the modifications to the proof of
\cite[Proposition~\ref{log-prop:flow-flow}]{BBS-saw4-log}.

By \cite[Theorem~\ref{flow-thm:flow}]{BBS-rg-flow} and
\cite[Corollary~\ref{flow-cor:masscont}]{BBS-rg-flow},
for any $(\mgen^2, \ggen_0) \in (0, \delta)^2$ and
$\tilde K_0 \in B_{\Wcal_0(\mgen^2, \ggen_0, \Zd)}(a_* \ggen_0^3)$,
there is a neighbourhood
${\sf N}(\ggen_0, \tilde K_0)$ of $(\ggen_0, \tilde K_0)$
such that for all
$(m^2, g_0, K_0) \in \Igen(\mgen^2) \times {\sf N}(\ggen_0, \tilde K_0)$,
there is an infinite-volume renormalisation group flow
\begin{equation}
(\Vch_j, K_j) = \xch^d_j(\mgen^2, \ggen_0, \tilde K_0; m^2, g_0, K_0)
\end{equation}
in \emph{transformed variables} $(\Vch_j, K_j)$.
The transformed variables are defined in
\cite[Section~\ref{log-sec:trans}]{BBS-saw4-log} and a flow
in the original variables can be recovered from the transformed flow.
The global solution is defined by
$\xch^c_j(m^2, g_0, K_0) = \xch^d_j(m^2, g_0, K_0; m^2, g_0, K_0)$
(or $\xch^c \equiv 0$ if $g_0 = 0$).
By \cite[Remark~\ref{flow-rk:Nrad}]{BBS-rg-flow},
the proof of regularity of $\xch^c$ can proceed as in \cite{BBS-saw4-log}.
The functions $(\nu_0^c, z_0^c)$ are given by the $(\nu_0, z_0)$ components
of $\xch^c_0 = (\Vch_0, K_0) = (V_0, K_0)$.
\end{proof}

\begin{rk}
The proof of \cite[Proposition~\ref{log-prop:flow-flow}]{BBS-saw4-log},
hence of Theorem~\ref{thm:flow-flow},
makes important use of the parameter $\ggen_0$ in order to prove regularity
of the renormalisation group flow in $g_0$. However, once the flow has been
constructed, we can and do set $\ggen_0 = g_0$.
\end{rk}

Suppose now that we are given some sufficiently small $\hat g_0 > 0$ and
a $\Lambda$-family $K_{0,\Lambda} \in \Wcal_0(m^2, \hat g_0, \Lambda)$
that satisfies
the bounds $\|K_{0,\Lambda}\|_{\Wcal_0(m^2, \hat g_0, \Lambda)} \le a_* \hat g_0^3$.
Then in any fixed volume $\Lambda = \Lambda_N$, we can generalise \eqref{e:Z0def} by defining
$Z_0 = (I_0 \circ K_0)(\Lambda)$ (\eqref{e:Z0def} is recovered when we set $K_0 = K^\pm_0$).
We also generalise \refeq{ZNdef} as $Z_N = \Ex_C\theta Z_0$, and 
let $\hat\chi_N(m^2, \hat g_0, K_0, \nu_0, z_0)$ be defined as in \eqref{e:chibarm}
from this $Z_N$ (generalising \eqref{e:chibarm}).
Then an analogue of \cite[Theorem~\ref{log-thm:suscept-diff}]{BBS-saw4-log}
(which corresponds to the case $K_0 = \1_\varnothing$)
follows from Theorem~\ref{thm:flow-flow}.
That is, if $(\nu_0^c, z_0^c) = (\nu_0^c(m^2, \hat g_0, K_0), z_0^c(m^2, \hat g_0, K_0))$, then
the limit $\hat\chi = \lim_{N\to\infty} \hat\chi_N$ exists anticipated
\begin{align}
\label{e:chi-m}
\hat\chi \left( m^2,\hat g_0,K_0, \nu_0^c,  z_0^c \right)
  &=
\frac{1}{m^2}, \\
\label{e:chiprime-m}
\ddp{\hat\chi}{\nu_0} \left(m^2,\hat g_0, K_0,\nu_0^c,  z_0^c \right)
  &\sim
- \frac{1}{m^4} \frac{c(\hat g_0^*, K_0)}{(\hat g_0^*\bubble_{m^2})^{1/4}}
  \quad \text{as $(m^2,\hat g_0) \to (0,\hat g_0^*)$},
\end{align}
where $c$ is a continuous function
and the \emph{bubble diagram} $\bubble_{m^2}$ is
is asymptotic to $(2\pi^2)^{-1} \log m^{-2}$, as $m^2 \downarrow 0$, when $d = 4$.
For instance, \eqref{e:chi-m} follows from \eqref{e:chibarm},
\eqref{e:ZVKN}, the bound on $K_N$ in
Theorem~\ref{thm:flow-flow},
and the bound on $W_N$ in \cite[\eqref{IE-e:W-logwish}]{BS-rg-IE}.
See \cite[Section~\ref{log-sec:suscept-diff-pf}]{BBS-saw4-log}
for details and for the proof of \eqref{e:chiprime-m}.

We wish to obtain a version of \eqref{e:chi-m}--\eqref{e:chiprime-m}
with the initial conditions of Section~\ref{sec:IK}, i.e.\ with
$(\hat g_0, K_0) = (g_0, K^+_0)$ (if $\gamma_0 > 0$) or $(\hat g_0, K_0) = (g_0 + 4d\gamma_0, K^-_0)$
(if $\gamma_0 < 0$).
It is straightforward to verify that $K^\pm_0 \in \Kcal_0$.
For instance, the fact that $K^\pm_0$ is supersymmetric
(which is required of all elements of $\Kcal_0$) follows
from the fact that $K^\pm_{0,x}$ is a function of $\tau_x$
(see \cite[Section~\ref{pt-sec:bulksym}]{BBS-rg-pt} for more on this topic).
It also follows from the definition that
the finite-volume coordinates $K^\pm_{0,\Lambda}$ form a $\Lambda$-family.

Moreover,
by Proposition~\ref{prop:KWcal}, if
$|\gamma_0|$ is sufficiently small (depending on $g_0$; we now take $\ggen_0=g_0$)
then $K_0 = K^\pm_0$ satisfies the bound required by Theorem~\ref{thm:flow-flow}.
However, we cannot apply the theorem immediately with this choice
of $K_0$,
due to the fact that $K^\pm_0$
depends on $(g_0, \nu_0, z_0)$.
We resolve this issue in the next section.

\subsection{Critical parameters}
\label{sec:nu0z0c}

For convenience, let
\begin{equation}
\lbeq{g0hatdef}
\hat g_0 = \hat g_0(g_0, \gamma_0) = g_0 + 4 d \gamma_0 \1_{\gamma_0 < 0}.
\end{equation}
Thus, $\hat g_0$ is the coefficient of $\tau_x^2$ in
$V^+_{0,x}$
when $\gamma_0 \ge 0$,
and in $V^-_{0,x}$
when $\gamma_0 < 0$.
Recall the function $K_0(g_0, \gamma_0, \nu_0, z_0)$
defined in \eqref{e:K0def}.
We wish to initialise the renormalisation group with $(\nu_0, z_0)$ a solution
to the system of equations
\begin{align}
&\nu_0 = \nu_0^c(m^2, \hat g_0(g_0, \gamma_0), K_0(g_0, \gamma_0, \nu_0, z_0)), \label{e:mu0c}
\\
&z_0 = z_0^c(m^2, \hat g_0(g_0, \gamma_0), K_0(g_0, \gamma_0, \nu_0, z_0)) \label{e:z0c}
.
\end{align}
Such a choice of $(\nu_0, z_0)$ will be critical for $K_0$,
where $K_0$ is itself evaluated at this same choice of $(\nu_0, z_0)$.

When $\gamma_0 = 0$, we get $K_0 = \1_\varnothing$, so $K_0$ no longer depends on $(\nu_0, z_0)$
and this system is solved by $(\nu_0^c(m^2, g_0, 0), z_0^c(m^2, g_0, 0))$
for any (small) $m^2, g_0 \geq 0$.
Local solutions for $\gamma_0 \neq 0$ can then be
constructed using a version of the implicit function theorem from \cite{LS14}
that allows for the continuous but non-smooth behaviour of $K_0$ in $\gamma_0$.
In order to obtain global solutions with certain desired regularity properties
(needed in the next section), we make use of Proposition~\ref{prop:IFT},
which is based on a version of the implicit function theorem from \cite{LS14}.

Suppose $\delta > 0$ and suppose $r : [0, \delta] \to [0, \infty)$
is a continuous positive-definite function; the latter
means that $r(x) > 0$ if $x > 0$ and $r(0) = 0$.
We define
\begin{equation}
\lbeq{Ddef}
D(\delta, r)
    =
\{ (w, x, y) \in [0, \delta]^2 \times (-\delta, \delta) : |y| \leq r(x) \}
\end{equation}
and we let $C^{0,1,\pm}(D(\delta, r))$ denote the space of continuous functions on $D(\delta, r)$
that are $C^1$ in $(x, y)$ away from $y = 0$, $C^1$ in $x$ everywhere,
and have left- and right-derivatives in $y$ at $y = 0$.
In our applications, we take $w = m^2$, $x = g_0$ or $\beta$,
and $y = \gamma_0$ or $\gamma$.

\begin{prop}
\label{prop:nuzhat}
There exists a continuous positive-definite function $\hat r : [0, \delta] \to [0, \infty)$
and continuous functions
$\hat\nu_0^c, \hat z_0^c \in C^{0,1,\pm}(D(\delta, \hat r))$ such that
the system \eqref{e:mu0c}--\eqref{e:z0c} is solved by $(\nu_0, z_0) = (\hat\nu_0^c, \hat z_0^c)$
whenever $(m^2, g_0, \gamma_0) \in D(\delta, \hat r)$.
Moreover, these functions satisfy the bounds
\begin{equation}
\label{e:hat-est}
\hat\nu_0^c = O(g_0),
\quad
\hat z_0^c = O(g_0)
\end{equation}
uniformly in $(m^2, \gamma_0)$.
\end{prop}

\begin{proof}
Recall the definition of $\ghat_0$ in \refeq{g0hatdef}, and
let
\begin{equation}
F(m^2, g_0, \gamma_0, \nu_0, z_0)
= (\nu_0, z_0)
  -
  (\nu_0^c(m^2, \hat g_0, K_0),
  z_0^c(m^2, \hat g_0, K_0)
),
\end{equation}
where $K_0 = K_0(g_0, \gamma_0, \nu_0, z_0)$.
Then for $\delta > 0$ small and an appropriate constant $c > 0$ (depending on $a_*$),
$F$ is well-defined on
\begin{equation}
\{ (m^2, g_0, \gamma_0, \nu_0, z_0) : (m^2, \hat g_0, \gamma_0) \in D(\delta, c g_0^3),
|\nu_0|, |z_0| \leq C_\DV g_0 \}.
\end{equation}
Indeed, for $(m^2, g_0, \gamma_0, \nu_0, z_0)$ in this domain,
Proposition~\ref{prop:KWcal} (with $\ggen_0 = g_0$) implies that $(m^2, \hat g_0, K_0)$ is in the domain of
$(\nu_0^c, z_0^c)$.
By Theorem~\ref{thm:flow-flow} and Proposition~\ref{prop:Ksmooth},
$F$ is $C^1$ in $(g_0, \nu_0, z_0)$
and also in $\gamma_0$ away from $\gamma_0 = 0$,
continuous in $m^2$, and has one-sided derivatives in $\gamma_0$ at $\gamma_0 = 0$.

For fixed $(\bar m^2, {\bar g_0}) \in [0, \delta]^2$,
set $({\bar\nu_0}, \bar z_0) = (\nu_0^c(\bar m^2, \bar g_0, 0), z_0^c(\bar m^2, \bar g_0, 0))$
so that
\begin{equation}
F(\bar m^2, \bar g_0, 0, \bar\nu_0, \bar z_0) = (0, 0).
\end{equation}
By \eqref{e:ddpK}, at $(\bar g_0, 0, \bar\nu_0, \bar z_0)$,
\begin{equation}
\frac{\partial K_{0,x}}{\partial\nu_0}
= \frac{\partial K_{0,x}}{\partial z_0} = 0.
\end{equation}
It follows that $D_{\nu_0,z_0} F(\bar m^2, \bar g_0, 0, \bar\nu_0, \bar z_0)$
is the identity map on $\R^2$.
The existence of $\delta, \hat r$ and $\hat\nu_0^c, \hat z_0^c$
follows from Proposition~\ref{prop:IFT} with
$w = m^2, x = g_0, y = \gamma_0, z = (\nu_0, z_0)$,
and with $r_1(g_0) = c g_0^3$, $r_2(g_0) = C_\DV g_0$.

By the fundamental theorem of calculus, for any $0 < a < \gamma_0$,
\begin{equation}
\hat\nu_0^c(m^2, g_0, \gamma_0)
  =
\hat\nu_0^c(m^2, g_0, a)
  +
\int_a^{\gamma_0} \ddp{\hat\nu_0^c}{\gamma_0} (m^2, g_0, t) \; dt.
\end{equation}
Taking the limit $a\downarrow 0$ and using \eqref{e:z0est}, we obtain
\begin{equation}
|\hat\nu_0^c(m^2, g_0, \gamma_0)|
  \leq
O(g_0)
  +
\gamma_0
\sup_{t \in (0, \gamma_0]}
\left|\ddp{\hat\nu_0^c}{\gamma_0}(m^2, g_0, t)\right|.
\end{equation}
The supremum above is bounded by a constant and so
the first estimate of \eqref{e:hat-est} for $\gamma_0 \geq 0$
follows from the fact that $|\gamma_0| \leq \hat r(g_0)$
(since $\hat r(g_0)$ can be taken as small as desired).
The case $\gamma_0 < 0$ and the second estimate follow similarly.
\end{proof}

\begin{cor}
\label{cor:rhatflow}
Fix $(m^2, g_0, \gamma_0) \in D(\delta, \hat r)$ with $g_0 > 0$ and
$m^2 \in [\delta L^{-2 (N - 1)}, \delta)$ and
set $(V_0, K_0) = (V^\pm_0, K^\pm_0)$
with $(\nu_0, z_0) = (\hat\nu_0^c, \hat z_0^c)$.
Then for any $N \in \N$,
there exists a sequence $(V_j, K_j) \in \domRG_j(m^2, g_0, \Lambda)$
such that
\begin{equation}
  \label{e:VjKjDj-hat}
  (V_{j+1},K_{j+1}) = (V_{j+1}(V_j, K_j), K_{j+1}(V_j, K_j)) \text{ for all } j < N
\end{equation}
and \eqref{e:ZjIjKj} is satisfied.
Moreover, the second-order evolution equation for $V_j$ is independent of $\gamma_0$.
\end{cor}

\begin{proof}
By Proposition~\ref{prop:KWcal},
and by taking $\hat r$ smaller if necessary,
$K_0 = K^\pm_0$ satisfies the estimate required by Theorem~\ref{thm:flow-flow}
whenever $(m^2, g_0, \gamma_0) \in D(\delta, \hat r)$. The
existence of the sequence \eqref{e:VjKjDj-hat} then follows from
Theorem~\ref{thm:flow-flow} and Proposition~\ref{prop:nuzhat}.
Although the presence of $\gamma_0$ causes a shift in initial
conditions, the second-order evolution of $V_j$ is still given by the map
$V_\pt$ (see \eqref{e:Vflow}),
which is independent of $\gamma_0$.
\end{proof}

By \eqref{e:chibarm}, $\hat\chi(m^2, g_0, \gamma_0, \nu_0, z_0) = \hat\chi(m^2, g_0, K_0, \nu_0, z_0)$,
where $K_0 = K_0(g_0, \gamma_0, \nu_0, z_0)$ is
defined in \eqref{e:K0def}. Then by \eqref{e:chi-m}--\eqref{e:chiprime-m},
Corollary~\ref{cor:rhatflow}, and \eqref{e:ddpK}, with $\hat g_0 = \hat g_0(g_0, \gamma_0)$, we have
\begin{align}
\label{e:chi-m-hat}
\hat\chi \left( m^2,\hat g_0,\gamma_0, \hat\nu_0^c, \hat z_0^c \right)
  &=
\frac{1}{m^2}, \\
\label{e:chiprime-m-hat}
\ddp{\hat\chi}{\nu_0} \left(m^2,\hat g_0, \gamma_0,\hat\nu_0^c, \hat z_0^c \right)
  &\sim
- \frac{1}{m^4} \frac{c(\hat g_0^*, \gamma_0)}{(\hat g_0^*\bubble_{m^2})^{1/4}}
  \quad \text{as $(m^2,g_0,\gamma_0) \to (0,g_0^*,\gamma_0^*)$},
\end{align}
where $\hat g_0^* = \hat g_0(g_0^*, \gamma_0^*)$ and we write $c(g_0, \gamma_0) = c(g_0, K_0)$.
Although \eqref{e:chiprime-m-hat} depends on $\gamma_0$, this dependence ultimately
only affects the computation of the critical point $\nu_c(\beta, \gamma)$ and the
constants $A_{\beta,\gamma}, B_{\beta,\gamma}$ in the proof of Theorem~\ref{thm:suscept}.
The asymptotic behaviour of the susceptibility in \eqref{e:chieps-asympt} results from the
logarithmic divergence of the bubble diagram $\bubble_{m^2}$ and the exponent $\frac{1}{4}$
that appears in the denominator in \eqref{e:chiprime-m-hat}.

\begin{rk}
We have invoked \eqref{e:ddpK} above in order to satisfy the condition
\begin{equation}
\|\partial K_0/\partial\nu_0\|_{\Wcal_0} \le O(g_0^3)
\end{equation}
required in the proof of \cite[Lemma~\ref{log-lem:gzmuprime}]{BBS-saw4-log}
(see \cite[\eqref{log-e:induct1}]{BBS-saw4-log}). This condition holds trivially
when $K_0$ does not depend on $\nu_0$, as in \eqref{e:chi-m}--\eqref{e:chiprime-m}.
\end{rk}


\subsection{Change of parameters}
\label{sec:nuztilde}

Recall from \eqref{e:chichibar} that
\begin{equation}
\label{e:chichihat}
\chi_N(\beta, \gamma, \nu)
  =
(1 + z_0) \hat\chi_N(m^2, g_0, \gamma_0, \nu_0, z_0),
\end{equation}
whenever the variables on the left- and right-hand sides satisfy
\begin{equation}
\label{e:gg0-re}
g_0 = (\beta - \gamma) (1 + z_0)^2,
\quad
\nu_0 = \nu (1 + z_0) - m^2,
\quad
\gamma_0 = \frac{1}{4d} \gamma (1 + z_0)^2.
\end{equation}
Given $\beta,\gamma,\nu$,
these relations leave free two of the variables
$(m^2, g_0, \gamma_0, \nu_0, z_0)$.
More generally, if any three of the variables
$(\beta, \gamma, \nu, m^2, g_0, \gamma_0, \nu_0, z_0)$
are fixed, then two of the remaining variables are free.
In the following two propositions,
which together form an extension of \cite[Proposition~\ref{log-prop:changevariables}]{BBS-saw4-log},
we fix three variables and show that the addition of the constraints
\begin{equation}
\label{e:crit-constraint}
\nu_0 = \hat \nu_0^c(m^2, g_0, \gamma_0),
  \quad
z_0   = \hat z_0^c(m^2, g_0, \gamma_0)
\end{equation}
allows us to uniquely specify the two remaining variables.
First, in Proposition~\ref{prop:changevariables1},
the three fixed variables are $(m^2, \beta, \gamma)$.

\begin{prop}
\label{prop:changevariables1}
There exist $\delta_* > 0$,
a continuous positive-definite function $r_* : [0, \delta_*] \to [0, \infty)$,
and continuous functions $(\nu^*, g_0^*, \gamma_0^*, \nu_0^*, z_0^*)$
defined for $(m^2, \beta, \gamma) \in D(\delta_*, r_*)$, such that
\eqref{e:gg0-re} and \eqref{e:crit-constraint} hold with $\nu = \nu^*$ and
$(g_0, \gamma_0, \nu_0, z_0) = (g^*_0, \gamma^*_0, \nu^*_0, z^*_0)$.
Moreover,
\begin{gather}
\label{e:gznustarbd}
g_0^* = \beta + O(\beta^2),
\quad
\nu_0^* = O(\beta),
\quad
z_0^* = O(\beta).
\end{gather}
\end{prop}

\begin{proof}
Suppose we have found the desired continuous functions $(g_0^*, \gamma_0^*)$
and that $g_0^*$ satisfies the first bound in \eqref{e:gznustarbd}.
Then the functions defined by
\begin{align}
&\nu_0^* = \hat\nu_0^c(m^2, g_0^*, \gamma_0^*), \quad
z_0^* = \hat z_0^c(m^2, g_0^*, \gamma_0^*), \quad
\nu^* = \frac{\nu_0^* + m^2}{1 + z_0^*}
\end{align}
are continuous, \eqref{e:gg0-re} is satisfied, and
the remaining bounds in
\eqref{e:gznustarbd}
follow using \eqref{e:hat-est}.

We first solve the third equation of \eqref{e:gg0-re}, and then solve the first equation
of \eqref{e:gg0-re}.
To this end, we begin by defining
\begin{equation}
f_1(m^2, g_0, \gamma, \gamma_0)
  =
\gamma_0 - (4d)^{-1} \gamma (1 + \hat z_0^c(m^2, g_0, \gamma_0))^2
\end{equation}
for $(m^2, g_0, \gamma_0) \in D(\delta, \hat r)$
and $|\gamma| \le \hat r(g_0)$
(recall that $\hat r$ is defined in Proposition~\ref{prop:nuzhat});
although $f_1$ is well-defined
for any $\gamma \in \R$, we restrict the domain in preparation
for our application of Proposition~\ref{prop:IFT}.
Note that $f_1$ is $C^1$ in $\gamma$ and
$f_1(\cdot, \cdot, \gamma, \cdot) \in C^{0,1,\pm}(D(\delta, \hat r))$ for any $\gamma$.
The equation $f_1(m^2, g_0, \gamma, \gamma_0) = 0$
has the solution $\gamma_0 = 0$ when $\gamma = 0$
and, for any $\gamma_0 \neq 0$,
\begin{equation}
\ddp{f_1}{\gamma_0}
  =
1 - (2d)^{-1} \gamma (1 + \hat z_0^c(m^2, g_0, \gamma_0)) \ddp{\hat z_0^c}{\gamma_0}.
\end{equation}
Since the one-sided $\gamma_0$ derivatives of $\hat z_0^c$ exist at $\gamma_0 = 0$,
we can see
that the $\gamma_0$ derivative of $f_1$ is well-defined
and equal to $1$ when $\gamma = 0$ for any small $\gamma_0$ (including $\gamma_0 = 0$).
Thus, by Proposition~\ref{prop:IFT}
(with $w = m^2$, $x = g_0$, $y = \gamma$, $z = \gamma_0$
and $r_1 = r_2 = \hat r$),
there exists a continuous function $\gamma^{(1)}_0(m^2, g_0, \gamma)$
on $D(\delta, r^{(1)})$ (for some continuous positive-definite function $r^{(1)}$ on $[0, \delta]$)
such that $f_1(m^2, g_0, \gamma, \gamma^{(1)}_0) = 0$.
Moreover, $\gamma^{(1)}_0$ is $C^1$ in $(g_0, \gamma)$.

Next, we define
\begin{equation}
f_2(m^2, \beta, \gamma, g_0)
  =
g_0 - (\beta - \gamma) (1 + \hat z_0^c(m^2, g_0, \gamma^{(1)}_0(m^2, g_0, \gamma)))^2
\end{equation}
for $(m^2, g_0, \gamma) \in D(\delta, r^{(1)})$ and $\beta \in [0, \delta_*]$,
where $\delta_* > 0$ will be made sufficiently small below.
Then $f_2(m^2, \beta, \gamma, g_0) = 0$ is solved by
$(\gamma, g_0) = (0, g_0^*(m^2, \beta, 0))$,
where $g_0^*(m^2, \beta, 0)$ was constructed in \cite[\eqref{log-e:ccstar2}]{BBS-saw4-log}.
By \cite[\eqref{log-e:gznustarbd}]{BBS-saw4-log}, $g_0^* = \beta + O(\beta^2)$,
so we may restrict the domain of $f_2$ so that $|g_0| \le 2 \beta$.
Moreover,
\begin{equation}
\ddp{f_2}{g_0}
  =
1 - 2 (\beta - \gamma) (1 + \hat z_0^c(m^2, g_0, \gamma^{(1)}_0))
\left( \ddp{\hat z_0^c}{g_0} + \ddp{\hat z_0^c}{\gamma_0} \ddp{\gamma^{(1)}_0}{g_0} \right).
\end{equation}
Differentiating both sides of
\begin{equation}
\gamma^{(1)}_0
  =
\frac{1}{4d} \gamma (1 + \hat z_0^c(m^2, g_0, \gamma^{(1)}_0))^2,
\end{equation}
and solving for $\ddp{\gamma^{(1)}_0}{g_0}$, gives
\begin{equation}
\ddp{\gamma^{(1)}_0}{g_0}
  =
\frac{\gamma (1 + \hat z_0^c)
  \ddp{\hat z_0^c}{g_0}}{2 d - \gamma (1 + \hat z_0^c) \ddp{\hat z_0^c}{\gamma_0}},
\end{equation}
where $\hat z_0^c$ and its derivatives are evaluated at $(m^2, g_0, \gamma^{(1)}_0)$.
Thus, $\ddp{\gamma^{(1)}_0}{g_0} = 0$ when $\gamma = 0$.
It follows that $\partial f_2/\partial g_0$
is well-defined when $(\gamma, g_0) = (0, g_0^*(m^2, \beta, 0))$ and equals
\begin{equation}
1 - 2 \beta (1 + \hat z_0^c(m^2, g_0^*, 0)) \ddp{\hat z_0^c}{g_0}(m^2, \beta, 0, g_0^*),
\end{equation}
which is positive when $\delta_*$ is small, by \eqref{e:hat-est}.
Thus, by Proposition~\ref{prop:IFT}
(with $w = m^2$, $x = \beta$, $y = \gamma$, $z = g_0$ and $r_1 = r^{(1)}$, $r_2(\beta) = 2\beta$),
there exists a function $g_0^*(m^2, \beta, \gamma) \in C^{0,1,\pm}(D(\delta_*, r^{(2)}))$
(for some continuous positive-definite function $r^{(2)}$ on $[0, \delta_*]$)
such that $f_2(m^2, \beta, \gamma, g_0^*) = 0$.

By the fact that $g_0^*$ solves $f_2 = 0$,
\begin{equation}
g_0^* = (\beta - \gamma) + O((\beta - \gamma)^2).
\end{equation}
Since $|\gamma| \le r^{(2)}(g_0)$ and $r^{(2)}(g_0)$ can be taken
as small as desired, this implies the first estimate in \eqref{e:gznustarbd}.
Thus, by taking $r_*$ sufficiently small, if $|\gamma| \le r_*(\beta)$, then
$|\gamma| \le r^{(2)}(g_0^*(m^2, \beta, \gamma))$.
Thus, for $\beta < \delta_*$ and $|\gamma| \leq r_*(\beta)$,
we can define
\begin{equation}
\gamma_0^*(m^2, \beta, \gamma) = \gamma^{(1)}_0(m^2, g_0^*(m^2, \beta, \gamma), \gamma),
\end{equation}
which completes the proof.
\end{proof}

Using Proposition~\ref{prop:changevariables1}, it is possible to
identify the critical point $\nu_c$, as follows.
By \eqref{e:chi-m-hat}, \eqref{e:chichihat}, Proposition~\ref{prop:finvol}, and Proposition~\ref{prop:changevariables1},
\begin{equation}
\label{e:chistar}
\chi(\beta, \gamma, \nu^*) = \frac{1 + z_0^*}{m^2} = \frac{1 + O(\beta)}{m^2}.
\end{equation}
Thus, with $\nu = \nu^*$, we see that $\chi < \infty$ when $m^2 > 0$, and
$\chi = \infty$ when $m^2 = 0$.
By \eqref{e:nuc-def}, this implies that
\begin{equation}
\label{e:nustarbd}
\nu_c(\beta, \gamma) = \nu^*(0, \beta, \gamma) = O(\beta),
  \quad
\nu_c(\beta, \gamma) < \nu^*(m^2, \beta, \gamma)
  \quad
(m^2 > 0).
\end{equation}
It follows that
\begin{equation}
\chi(\beta, \gamma, \nu_c) = \infty,
\end{equation}
which is a fact that cannot be concluded immediately from the definition \refeq{nuc-def}.

In \eqref{e:chistar}, $\chi$ is evaluated at $\nu^* = \nu^*(m^2, \beta, \gamma)$.
However, in the setting of Theorem~\ref{thm:suscept},
we need to evaluate $\chi$ at a \emph{given} value of $\nu$
and then take $\nu \downarrow \nu_c$.
To do so, we must determine a choice of $m^2$ in terms of $\nu$
such that \eqref{e:gg0-re} is satisfied and this choice
must approach $0$ (as it should by \eqref{e:nustarbd})
right-continuously as $\nu\downarrow\nu_c$.
The following proposition carries out this construction.
In the following, the functions $\tilde m^2, \tilde g_0$ should not be
confused with the parameter $\mgen^2, \ggen_0$ that appeared previously
in the $\Wcal_j$ norms.

\begin{prop}
\label{prop:changevariables2}
Write $\nu = \nu_c + \varepsilon$.
There exist functions $\tilde m^2, \tilde g_0, \tilde\gamma_0, \tilde\nu_0, \tilde z_0$
of $(\varepsilon, \beta, \gamma) \in D(\delta_*, r_*)$
(all right-continuous as $\varepsilon\downarrow 0$)
such that \eqref{e:gg0-re} and \eqref{e:crit-constraint} hold with
\begin{equation}
(m^2, g_0, \gamma_0, \nu_0, z_0) = (\tilde m^2, \tilde g_0, \tilde\gamma_0, \tilde\nu_0, \tilde z_0).
\end{equation}
Moreover,
\begin{gather}
\label{e:mtildebd}
\tilde m^2(0, \beta, \gamma) = 0,
    \qquad
\tilde m^2(\varepsilon, \beta, \gamma) > 0
    \quad
(\varepsilon > 0). \\
\label{e:gznutildebd}
\tilde g_0 = \beta + O(\beta^2),
    \quad
\tilde \nu_0 = O(\beta),
    \quad
\tilde z_0 = O(\beta).
\end{gather}
\end{prop}

\begin{proof}
The proof is a minor modification of the proof in \cite{BBS-saw4-log},
using Proposition~\ref{prop:changevariables1}.
Define
\begin{equation}
\label{e:mtildef}
\tilde m^2
    =
\tilde m^2 (\varepsilon,\beta,\gamma)
    =
\inf \{m^2 > 0 : \nu^*(m^2, \beta, \gamma) = \nu_c(\beta, \gamma) + \varepsilon \},
\end{equation}
on $D(\delta_*, r_*)$.
By continuity of $\nu^*$,
the infimum is attained and
\begin{equation}
\nu_c(\beta, \gamma) + \varepsilon = \nu^*(\tilde m^2(\varepsilon, \beta, \gamma), \beta, \gamma).
\end{equation}
From the above expression, continuity of $\nu^*$, and \eqref{e:nustarbd},
it follows that $\tilde m^2$ is right-continuous as $\varepsilon\downarrow 0$.
It is immediate that \eqref{e:mtildebd} holds.
Also, the functions of $(\varepsilon,\beta,\gamma)$ defined by
\begin{align}
&\tilde\nu_0 = \nu_0^*(\tilde m^2, \beta, \gamma), \quad
\tilde z_0 = z_0^*(\tilde m^2, \beta, \gamma), \\
&\tilde g_0 = (\beta - \gamma) (1 + \tilde z_0)^2 ,\quad
\tilde\gamma_0 = \frac{1}{4d} \gamma (1 + \tilde z_0)^2
\end{align}
are right-continuous as $\varepsilon \downarrow 0$ and satisfy \eqref{e:gg0-re}.
The bounds \eqref{e:gznutildebd} follow from the definitions
and \eqref{e:gznustarbd}, and the proof is complete.
\end{proof}

\subsection{Conclusion of the argument}
\label{sec:conclusion}

By \eqref{e:chi-m-hat}, \eqref{e:chichihat}, Proposition~\ref{prop:finvol},
and Proposition~\ref{prop:changevariables2}
\begin{equation}
\chi(\beta, \gamma, \nu)
  =
\frac{1 + \tilde z_0}{\tilde m^2}.
\end{equation}
Using this, \eqref{e:chi-m-hat}, and \eqref{e:chiprime-m-hat},
by exactly the same argument as in \cite[Section~\ref{log-sec:pfsuscept}]{BBS-saw4-log},
there is
a differential relation between $\ddp{\chi}{\nu}$ and $\chi$,
whose solution yields Theorem~\ref{thm:suscept}(ii).

The reason the susceptibility is handled first is that its leading-order
critical behaviour can be computed from the second-order flow of the \emph{bulk} coupling
constants $(g_j, \nu_j, z_j)$. In contrast, in order to study the two-point
function, we begin by writing
\begin{equation}
\label{e:phi-sigma}
\bar\phi_a \phi_b
    =
\ddp{^2}{\sigma_a\partial\sigma_b} e^{\sigma_a\bar\phi_a + \sigma_b\phi_b}
\Big|_{\sigma_a=\sigma_b=0}
\end{equation}
in \eqref{e:GG2}. The incorporation of the exponential function
$e^{\sigma_a\bar\phi_a + \sigma_b\phi_b}$
into $Z_0$
is equivalent to subtracting
\begin{equation}
\sigma_a \bar\phi_a \1_{x=a} + \sigma_b \bar\phi_b \1_{x=b}
\end{equation}
from $V^\pm_0$. The renormalisation group map now acts on a polynomial of the form
\begin{equation}
g_j \tau^2 + \nu_j \tau + z_j \tau_\Delta
    - \lambda_{a,j} \sigma_a \bar\phi_a \1_{x=a}
    - \lambda_{b,j} \sigma_b \phi_b \1_{x=b}
    - \frac{1}{2} \sigma_a \sigma_b (q_{a,j} \1_{x=a} + q_{b,j} \1_{x=b}).
\end{equation}
We have only included terms up to second order in $(\sigma_a, \sigma_b)$ because,
by \eqref{e:phi-sigma}, only these are needed to study the two-point function.
The coefficients $(\lambda_{a,j}, \lambda_{b,j}, q_{a,j}, q_{b,j})$
are referred to as \emph{observable} coupling constants and the behaviour of these
coupling constants under the action of the renormalisation group is studied in detail
in \cite{BBS-saw4,ST-phi4}.

It was shown in \cite{BBS-saw4} that the observable flow does not affect the bulk flow.
Moreover, the second-order evolution of the observable
flow remains identical to that of the case $\gamma_0 = 0$.
This occurs for the same reason that the bulk flow is unaffected to second
order by $\gamma_0$ (as in the statement of Corollary~\ref{cor:rhatflow}):
namely, the second-order contributions to the observable flow are produced
by an extension of the map $\Vpt$ (recall \eqref{e:Vflow}), whose definition
does not depend on $\gamma_0$.
Thus, the analysis of the observable flow when $\gamma_0$
is small can proceed in the same way as when $\gamma_0 = 0$.
That is, the same analysis that was carried out in \cite{BBS-saw4}
to study the two-point function applies directly
here to prove Theorem~\ref{thm:suscept}(i).

The analysis
of the correlation length of order $p$
in \cite{BSTW-clp} also applies directly here, and for the same reason: the second-order
flow of coupling constants is independent of $\gamma_0$.  This gives Theorem~\ref{thm:suscept}(iii).

\subsection{A version of the implicit function theorem}
\label{sec:IFT}

We make use of \cite[Chapter 4, Theorem~9.3]{LS14},
which is a version of the implicit function theorem that allows
for a continuous, rather than differentiable, parameter.
While the precise statement of \cite[Chapter 4, Theorem~9.3]{LS14}
takes this parameter from an open subset of a Banach space,
by \cite[Chapter 4, Theorem~9.2]{LS14}, the
parameter can in fact be taken from an arbitrary metric space.
With this minor change, we restate \cite[Chapter 4, Theorem~9.3]{LS14}
as the following proposition.

\begin{prop}
\label{prop:LS-IFT}
Let $A$ be a metric space, let $W,X$ be Banach spaces,
and let $B \subset W$ be an open subset.
Let $F : A \times B \to X$ be continuous,
and suppose that $F$ is $C^1$ in its second argument.
Let $(\alpha, \beta) \in A \times B$ be a point such that
$F(\alpha, \beta) = 0$ and $D_2 F(\alpha, \beta)^{-1}$ exists.
Then there are open balls $M \ni \alpha$ and $N \ni \beta$
and a unique continuous mapping $f : M \to N$ such that $F(\xi, f(\xi)) = 0$
for all $\xi \in M$.
\end{prop}

We also use the following lemma, which is a small modification of
\cite[Chapter 3, Theorem~11.1]{LS14}. In particular, it considers
functions that may only be left- or right-differentiable.

\begin{lemma}
\label{lem:IFT-C1}
Let $F$ be a mapping as in the previous proposition
with $A \subset \R^{m_1} \times \R^{m_2}$.
In addition, suppose that $F$ is left-differentiable (respectively, right-differentiable)
in $\alpha_2$ at $(\alpha, \beta)$, with $\alpha = (\alpha_1, \alpha_2)$.
If $f$ is a continuous mapping defined in a neighbourhood of
$\alpha$, such that $F(\xi, f(\xi)) = 0$,
then $f$ is left-differentiable (respectively, right-differentiable) in $\alpha_2$ at $\alpha$.
\end{lemma}

The above results lead to the following proposition, which we apply
in the proofs of Propositions~\ref{prop:nuzhat} and \ref{prop:changevariables1}.
Recall that $D(\delta, r)$ is defined in \refeq{Ddef}.

\begin{prop}
\label{prop:IFT}
Let $\delta > 0$, and let $r_1, r_2$ be continuous positive-definite functions on $[0, \delta]$.
Set
\begin{equation}
    D(\delta, r_1, r_2)
    =
    \{ (w, x, y, z) \in D(\delta, r_1) \times \R^n : |z| \leq r_2(x) \},
\end{equation}
and let $F$ be a continuous function on $D(\delta, r_1, r_2)$ that is $C^1$ in $(x, z)$.
Suppose that for all $(\bar w, \bar x) \in [0, \delta]^2$ there exists $\bar z$
such that both $F(\bar w, \bar x, 0, \bar z) = 0$
and $D_Y F(\bar w, \bar x, 0, \bar z)$ is invertible.
Then there is a continuous positive-definite function $r$ on $[0, \delta]$ and
a continuous map $f : D(\delta, r) \to \R^n$
that is $C^1$ in $x$
and such that $F(w, x, y, f(w, x, y)) = 0$
for all $(w, x, y) \in D(\delta, r)$.
Moreover, if $F$ is left-differentiable
(respectively, right-differentiable) in $y$ at some point $(w, x, y, z)$,
then $f$ is left-differentiable (respectively, right-differentiable) at $(w, x, y)$.
\end{prop}

\begin{proof}
Take any $(\bar w, \bar x) \in [0, \delta] \times (0, \delta]$
and let $R(\bar w, \bar x)$ be the maximal radius $s$ such that
for all $(w, x, y) \in B(\bar w, \bar x, 0; s)$ there exists $z$
such that both $F(w, x, y, z) = 0$ and $D_Z F(w, x, y, z)$ is
invertible. By continuity of $(D_Z F(w, x, y, z))^{-1}$ near
$(\bar w, \bar x, 0, \bar z)$, and by Proposition~\ref{prop:LS-IFT}
(applied to the restriction of $F$ to $A \times B$, for some
$A \ni (\bar w, \bar x, 0)$ and an open set $B \ni \bar z$),
we have $R(\bar w, \bar x) > 0$ and there is a continuous function
\begin{equation}
f_{\bar w,\bar x} : B(\bar w, \bar x, 0; R(\bar w, \bar x)) \to \R^n
\end{equation}
such that $F(w, x, y, f_{\bar w,\bar x}(w, x, y)) = 0$
for all $(w, x, y) \in B(\bar w, \bar x, 0; R(\bar w, \bar x))$.
Moreover, the unique solution to $F(w, x, y, z) = 0$
is given by $z = f_{\bar w,\bar x}(w, x, y)$ for all such $(w, x, y)$.
By an application of Lemma~\ref{lem:IFT-C1}
(with $\alpha_1 = (w, x), \alpha_2 = y$),
we see that $f_{\bar w, \bar x}$ is
left- or right-differentiable in $y$ wherever $F$ is.
By another application of Lemma~\ref{lem:IFT-C1} (with $\alpha_1 = (w, y), \alpha_2 = x$),
we see that $f_{\bar w, \bar x}$ is $C^1$ in $x$.

Set $R(\bar w, 0) = 0$ for all $\bar w \in [0, \delta]$, and
let
\begin{equation}
D_f = \bigcup_{(\bar w,\bar x)\in [0, \delta]^2} B(\bar w, \bar x, 0; R(\bar w, \bar x)).
\end{equation}
We define $f(w, 0, 0) = 0$ and, for $x > 0$,
\begin{equation}
f(w, x, y) = f_{\bar w,\bar x}(w, x, y)
  \quad\text{for}\quad
(w, x, y) \in B(\bar w, \bar x, 0; R(\bar w, \bar x)).
\end{equation}
By uniqueness, this function is well-defined.
Continuity of $f$ at $(w, 0, 0)$
follows from the fact that $|f(w, x, y)| \le r_2(x)$.
The remaining desired regularity properties of $f$
follow from those of the $f_{\bar w,\bar x}$.
It remains to show that $D(\delta,r) \subset D_f$
for some continuous positive-definite function $r$ on $[0, \delta]$.

First, let us show that $R$ is continuous on $[0, \delta]^2$.
Let $\bar x > 0$ and fix $0 < \epsilon < R(\bar w, \bar x)$.
Then for any $(\bar w', \bar x') \in [0,\delta] \times (0, \delta]$ such that
$|(\bar w, \bar x) - (\bar w', \bar x')| < \epsilon$,
we have $B(\bar w', \bar x', 0; R(\bar w, \bar x) - \epsilon) \subset B(\bar w, \bar x, 0; R(\bar w, \bar x))$
by maximality of $R$.
It follows that $R(\bar w', \bar x') \geq R(\bar w, \bar x) - \epsilon$.
By a similar argument, $R(\bar w', \bar x') \leq R(\bar w, \bar x) + \epsilon$,
so $|R(\bar w, \bar x) - R(\bar w', \bar x')| \leq \epsilon$.
Thus, $R$ is continuous on $[0, \delta] \times (0, \delta]$.
Continuity at $\bar x = 0$ follows from the fact that $R(\bar w, \bar x) \le r_1(\bar x)$
uniformly in $\bar w$.

For $\bar x \in [0,\delta]$, let
\begin{equation}
r(\bar x) = \inf (R(\bar w, \bar x) : \bar w \in [0, \delta]).
\end{equation}
Since $R(\cdot, \bar x)$ is continuous, $r(\bar x) > 0$
for $\bar x > 0$. Moreover, $0 \le r(0) \le r_1(0) = 0$, so $r$ is positive-definite.
Continuity of $r$ follows from joint continuity of $R$.
For any $(w, x, y) \in D(\delta, r)$ (with this choice of $r$),
\begin{equation}
|(w, x, y) - (w, x, 0)| = |y| < r(x) \leq R(w, x),
\end{equation}
so $(w, x, y) \in B(w, x, 0; R(w, x))$.
We conclude that $D(\delta, r) \subset D_f$.
\end{proof}


\section*{Acknowledgements}

The work of RB was supported in part by the Simons Foundation.
The work of GS and BCW was supported in part by NSERC of
Canada.
We thank the referees for useful suggestions.


\end{document}

%% file: macros.tex
\usepackage{amsfonts}
\usepackage{amsmath,amssymb,amsthm}
\usepackage{appendix}
\usepackage{bbm} 
\usepackage{amsbsy}
\usepackage{enumerate}
\usepackage{cite}
\usepackage{enumerate}

\InputIfFileExists{./macros_local.tex}{}{}

\ifdefined\macrosPa
  \usepackage[textwidth=465pt,textheight=650pt,centering]{geometry} 
\else\ifdefined\macrosPb
  \usepackage[textwidth=500pt,textheight=650pt,centering]{geometry} 
\fi\fi

\ifdefined\macrosS


  \usepackage{mathptmx}
  \DeclareMathAlphabet{\mathcal}{OMS}{cmsy}{m}{n}
\fi

\ifdefined\macrosBirk
\else
\usepackage[dvips]{graphicx}
\fi

\ifdefined\macrosSB
\input ../def99c 
\else
\input def99c 
\fi

\UseSection   
\usepackage[usenames]{color}

\definecolor{bw}{RGB}{240, 120, 0}
\definecolor{at}{rgb}{0.0, 0.5, 0.0} 

\newcommand{\shift}{\!\!\!\!}

\newcommand{\DV}{\Dcal}

\renewcommand{\to} {\rightarrow}

\newcommand{\R}{\Rbold}
\newcommand{\Z}{\Zbold}

\newcommand{\N}{\Nbold}
\newcommand{\C}{\mathbb{C}}

\newcommand{\1}{\mathbbm{1}}

\newcommand{\psib}{\bar\psi}

\newcommand{\Ex}{\mathbb{E}}

\newcommand{\chicCov}{{\chi}}

\newcommand{\lt}{\ell}

\newcommand{\bubble}{{\sf B}}

\newcommand{\diam}[1]{\textrm{diam}(#1)}

\newcommand{\pt}{{\rm pt}}

\newcommand{\Vpt}{V_{\rm pt}}

\newcommand{\xch}{\check{x}}

\newcommand{\Vch}{\check{V}}

\newcommand{\h}{\mathfrak{h}}
\ifdefined\macrosSB \else

\fi

\renewcommand{\ghat}{\hat{g}}
\newcommand{\ggen}{\tilde{g}}

\newcommand{\mgen}{\tilde{m}}
\newcommand{\Iint}{\mathbb{I}}
\newcommand{\Igen}{\tilde{\mathbb{I}}}

\newcommand{\domRG}{\mathbb{D}}

\newcommand{\ddp}[2]{\frac{\partial #1}{\partial #2}}

\newcommand{\phib}{\bar\phi}

\newcommand{\amain} {a}

\ifdefined\macrosH
  \usepackage{xr-hyper}
  \usepackage{hyperref}
  \hypersetup{hypertexnames=false}
  \hypersetup{colorlinks,citecolor=blue,linkcolor=blue}  

\else\ifdefined\macrosHarxiv
  \usepackage{xr-hyper}
  \usepackage{hyperref}
  \hypersetup{hypertexnames=false}

\else
  \newcommand{\texorpdfstring}[2]{#1}
  \usepackage{xr}
\fi\fi

%% file: def99c.tex
\def\UseSection{
        \numberwithin{equation}{section}
	\theoremstyle{plain}
        \newtheorem{theorem}    {Theorem}[section]
        \DefineTheorems 
}

\def\DefineTheorems{
	
	\newtheorem{lemma}      [theorem] {Lemma}
	
	\newtheorem{prop}       [theorem] {Proposition}
	
	\newtheorem{cor}        [theorem] {Corollary}

	\theoremstyle{definition}
	\newtheorem{defn}       [theorem] {Definition}

	\newtheorem{rk} 	[theorem] {Remark}
	\theoremstyle{definition}

}

\newcommand{\bt}   {\begin{theorem}}
\newcommand{\et}   {\end  {theorem}}
\newcommand{\bl}   {\begin{lemma}}
\newcommand{\el}   {\end  {lemma}}
\newcommand{\bp}   {\begin{prop}}
\newcommand{\ep}   {\end  {prop}}
\newcommand{\bc}   {\begin{cor}}
\newcommand{\ec}   {\end  {cor}}
\newcommand{\bd}   {\begin{defn}}
\newcommand{\ed}   {\end  {defn}}

\newcommand{\ba}   {\begin{array}}
\newcommand{\ea}   {\end  {array}}
\newcommand{\be}   {\begin{enumerate}}
\newcommand{\ee}   {\end  {enumerate}}
\newcommand{\bi}   {\begin{itemize}}
\newcommand{\ei}   {\end  {itemize}}

\def\eq#1\en{\begin{equation}#1\end{equation}}  
\def\eqsplit#1\ensplit{
	\begin{equation}\begin{split}#1\end{split}\end{equation}
	}
\def\eqalign#1\enalign{
	\begin{align}#1\end{align}
	}
\def\eqmul#1\enmul{
	\begin{multline}#1\end{multline}
	}
\newcommand{\eqarrstar} {\begin{eqnarray*}} 
\newcommand{\enarrstar} {\end{eqnarray*}} 
\newcommand{\eqarray}   {\begin{eqnarray}} 
\newcommand{\enarray}   {\end{eqnarray}}

\newcommand{\lbeq}[1]  {\label{e:#1}}
\newcommand{\refeq}[1] {\eqref{e:#1}}    

\makeatletter
\newcommand{\labelcounter}[2]{{%
	\stepcounter{#1}
	\protected@write\@auxout{}%
	{\string\newlabel{#2}{{\csname the#1\endcsname}{\thepage}}}%
	{\ref{#2}}
	}}
\makeatother
\newcommand{\Nbold} {{\mathbb N}}

\newcommand{\Rbold} {{\mathbb R}}

\newcommand{\Zbold} {{\mathbb Z}}

\newcommand{\Bcal}   {\mathcal{B}} 
\newcommand{\Ccal}   {\mathcal{C}} 
\newcommand{\Dcal}   {\mathcal{D}} 
 
\newcommand{\Fcal}   {\mathcal{F}} 
\newcommand{\Gcal}   {\mathcal{G}}

\newcommand{\Kcal}   {\mathcal{K}}

\newcommand{\Ncal}   {\mathcal{N}} 
 
\newcommand{\Pcal}   {\mathcal{P}}

\newcommand{\Ucal}   {\mathcal{U}} 
 
\newcommand{\Wcal}   {\mathcal{W}}

\newcommand{\ghat}  {{ \hat{g}  }}

\newcommand{\Rd}    {{ {\Rbold}^d}}
\newcommand{\Zd}    {{ {\Zbold}^d }}

\newcommand{\spose}[1] {{\hbox to 0pt{#1\hss}} }
\newcommand{\ltapprox} {\mathrel{\spose{\lower 3pt\hbox{$\mathchar"218$}}
 \raise 2.0pt\hbox{$\mathchar"13C$}}}
\newcommand{\gtapprox} {\mathrel{\spose{\lower 3pt\hbox{$\mathchar"218$}}
 \raise 2.0pt\hbox{$\mathchar"13E$}}}



%% file: polymer-contact.pspdftex
\begin{picture}(0,0)%
\includegraphics{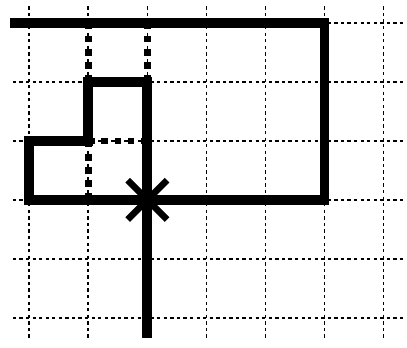}%
\end{picture}%
\setlength{\unitlength}{4144sp}%
\begingroup\makeatletter\ifx\SetFigFont\undefined%
\gdef\SetFigFont#1#2#3#4#5{%
  \reset@font\fontsize{#1}{#2pt}%
  \fontfamily{#3}\fontseries{#4}\fontshape{#5}%
  \selectfont}%
\fi\endgroup%
\begin{picture}(1856,1586)(1847,1245)
\end{picture}%

%% file: polymer-phasediagram.pspdftex
\begin{picture}(0,0)%
\includegraphics{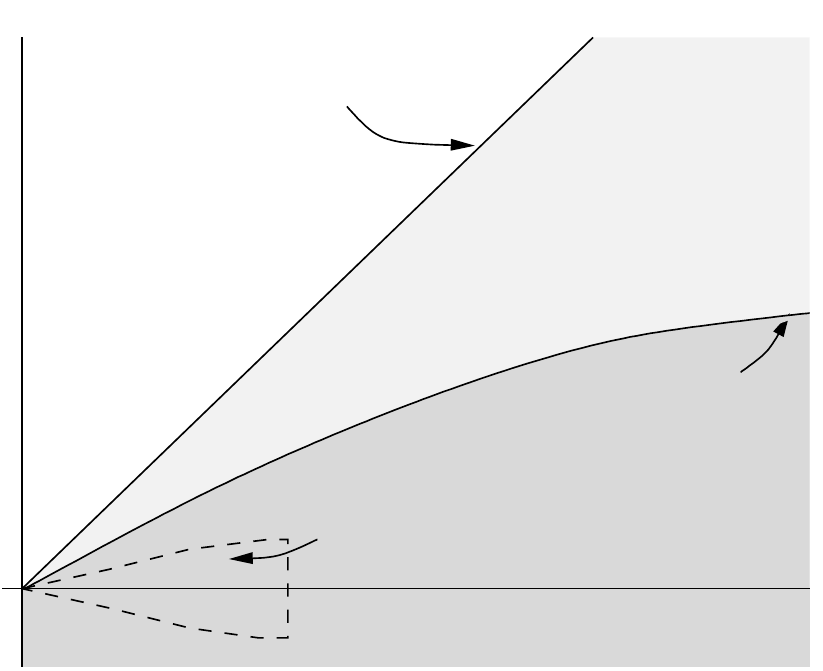}%
\end{picture}%
\setlength{\unitlength}{4144sp}%
\begingroup\makeatletter\ifx\SetFigFont\undefined%
\gdef\SetFigFont#1#2#3#4#5{%
  \reset@font\fontsize{#1}{#2pt}%
  \fontfamily{#3}\fontseries{#4}\fontshape{#5}%
  \selectfont}%
\fi\endgroup%
\begin{picture}(3762,3060)(5479,-4393)
\put(5581,-1456){\makebox(0,0)[b]{\smash{{\SetFigFont{9}{10.8}{\familydefault}{\mddefault}{\updefault}{\color[rgb]{0,0,0}$\gamma$}%
}}}}
\put(8326,-2536){\makebox(0,0)[b]{\smash{{\SetFigFont{9}{10.8}{\familydefault}{\mddefault}{\updefault}{\color[rgb]{0,0,0}$\bar\nu=1/d$}%
}}}}
\put(7111,-1726){\makebox(0,0)[b]{\smash{{\SetFigFont{9}{10.8}{\familydefault}{\mddefault}{\updefault}{\color[rgb]{0,0,0}$\bar\nu=1/(1+d)$}%
}}}}
\put(6301,-2626){\makebox(0,0)[b]{\smash{{\SetFigFont{9}{10.8}{\familydefault}{\mddefault}{\updefault}{\color[rgb]{0,0,0}$\bar\nu=0$}%
}}}}
\put(8596,-3076){\makebox(0,0)[b]{\smash{{\SetFigFont{9}{10.8}{\familydefault}{\mddefault}{\updefault}{\color[rgb]{0,0,0}$\bar\nu=\bar\nu_\theta$}%
}}}}
\put(8371,-3526){\makebox(0,0)[b]{\smash{{\SetFigFont{9}{10.8}{\familydefault}{\mddefault}{\updefault}{\color[rgb]{0,0,0}$\bar\nu=\bar\nu_{\mathrm{SAW}}$}%
}}}}
\put(9226,-4066){\makebox(0,0)[lb]{\smash{{\SetFigFont{9}{10.8}{\familydefault}{\mddefault}{\updefault}{\color[rgb]{0,0,0}$\beta$}%
}}}}
\put(6976,-3796){\makebox(0,0)[lb]{\smash{{\SetFigFont{9}{10.8}{\familydefault}{\mddefault}{\updefault}{\color[rgb]{0,0,0}{our result}}%
}}}}
\end{picture}%

%% file: saw-sa.bbl
\begin{thebibliography}{10}

\bibitem{Baue13a}
R.~Bauerschmidt.
\newblock A simple method for finite range decomposition of quadratic forms and
  {Gaussian} fields.
\newblock {\em Probab. Theory Related Fields}, {\bf 157}:817--845, (2013).

\bibitem{BBS-saw4}
R.~Bauerschmidt, D.C. Brydges, and G.~Slade.
\newblock Critical two-point function of the 4-dimensional weakly self-avoiding
  walk.
\newblock {\em Commun.\ Math.\ Phys.}, {\bf 338}:169--193, (2015).

\bibitem{BBS-saw4-log}
R.~Bauerschmidt, D.C. Brydges, and G.~Slade.
\newblock Logarithmic correction for the susceptibility of the 4-dimensional
  weakly self-avoiding walk: a renormalisation group analysis.
\newblock {\em Commun.\ Math.\ Phys.}, {\bf 337}:817--877, (2015).

\bibitem{BBS-rg-pt}
R.~Bauerschmidt, D.C. Brydges, and G.~Slade.
\newblock A renormalisation group method. {III}. {Perturbative} analysis.
\newblock {\em J. Stat. Phys}, {\bf 159}:492--529, (2015).

\bibitem{BBS-rg-flow}
R.~Bauerschmidt, D.C. Brydges, and G.~Slade.
\newblock Structural stability of a dynamical system near a non-hyperbolic
  fixed point.
\newblock {\em Ann. Henri Poincar\'e}, {\bf 16}:1033--1065, (2015).

\bibitem{BSTW-clp}
R.~Bauerschmidt, G.~Slade, A.~Tomberg, and B.C. Wallace.
\newblock Finite-order correlation length for 4-dimensional weakly
  self-avoiding walk and $|\varphi|^4$ spins.
\newblock {\em Annales Henri Poincar\'e}, (2016).
\newblock Online first.

\bibitem{BEI92}
D.~Brydges, S.N. Evans, and J.Z. Imbrie.
\newblock Self-avoiding walk on a hierarchical lattice in four dimensions.
\newblock {\em Ann. Probab.}, {\bf 20}:82--124, (1992).

\bibitem{BGM04}
D.C. Brydges, G.~Guadagni, and P.K. Mitter.
\newblock Finite range decomposition of {Gaussian} processes.
\newblock {\em J. Stat. Phys.}, {\bf 115}:415--449, (2004).

\bibitem{BIS09}
D.C. Brydges, J.Z. Imbrie, and G.~Slade.
\newblock Functional integral representations for self-avoiding walk.
\newblock {\em Probab.\ Surveys}, {\bf 6}:34--61, (2009).

\bibitem{BM91}
D.C. Brydges and I.~Mu\~noz Maya.
\newblock An application of {Berezin} integration to large deviations.
\newblock {\em J. Theoret. Probab.}, {\bf 4}:371--389, (1991).

\bibitem{BS-rg-norm}
D.C. Brydges and G.~Slade.
\newblock A renormalisation group method. {I}. {Gaussian} integration and
  normed algebras.
\newblock {\em J. Stat. Phys}, {\bf 159}:421--460, (2015).

\bibitem{BS-rg-loc}
D.C. Brydges and G.~Slade.
\newblock A renormalisation group method. {II}. {Approximation by local
  polynomials}.
\newblock {\em J. Stat. Phys}, {\bf 159}:461--491, (2015).

\bibitem{BS-rg-IE}
D.C. Brydges and G.~Slade.
\newblock A renormalisation group method. {IV}. {Stability} analysis.
\newblock {\em J. Stat. Phys}, {\bf 159}:530--588, (2015).

\bibitem{BS-rg-step}
D.C. Brydges and G.~Slade.
\newblock A renormalisation group method. {V}. {A} single renormalisation group
  step.
\newblock {\em J. Stat. Phys}, {\bf 159}:589--667, (2015).

\bibitem{BS85}
D.C. Brydges and T.~Spencer.
\newblock Self-avoiding walk in 5 or more dimensions.
\newblock {\em Commun. Math. Phys.}, {\bf 97}:125--148, (1985).

\bibitem{Clis10}
N.~Clisby.
\newblock Accurate estimate of the critical exponent $\nu$ for self-avoiding
  walks via a fast implementation of the pivot algorithm.
\newblock {\em Phys. Rev. Lett.}, {\bf 104}:055702, (2010).

\bibitem{HH17}
A.~Hammond and T.~Helmuth.
\newblock Self-avoiding walk with a self-attraction.
\newblock In preparation.

\bibitem{Hara08}
T.~Hara.
\newblock Decay of correlations in nearest-neighbor self-avoiding walk,
  percolation, lattice trees and animals.
\newblock {\em Ann. Probab.}, {\bf 36}:530--593, (2008).

\bibitem{HS92a}
T.~Hara and G.~Slade.
\newblock Self-avoiding walk in five or more dimensions. {I.} {The} critical
  behaviour.
\newblock {\em Commun.\ Math.\ Phys.}, {\bf 147}:101--136, (1992).

\bibitem{Helm16}
T.~Helmuth.
\newblock Loop-weighted walk.
\newblock {\em Ann. Inst. Henri Poincar\'{e} Comb. Phys. Interact.}, {\bf
  3}:55--119, (2016).

\bibitem{HK01a}
R.~van~der Hofstad and A.~Klenke.
\newblock Self-attractive random polymers.
\newblock {\em Ann. Appl. Probab.}, {\bf 11}:1079--1115, (2001).

\bibitem{HKK02}
R.~van~der Hofstad, A.~Klenke, and W.~K{\"o}nig.
\newblock The critical attractive random polymer in dimension one.
\newblock {\em J. Stat. Phys.}, {\bf 106}:477--520, (2002).

\bibitem{Holl09}
F.~den Hollander.
\newblock {\em Random Polymers}.
\newblock Springer, Berlin, (2009).
\newblock Lecture Notes in Mathematics Vol. 1974. Ecole d'Et\'{e} de
  Probabilit\'{e}s de Saint--Flour XXXVII--2007.

\bibitem{Jans15}
E.~J. Janse~van Rensburg.
\newblock {\em The Statistical Mechanics of Interacting Walks, Polygons,
  Animals and Vesicles}.
\newblock Oxford University Press, Oxford, 2nd edition, (2015).

\bibitem{LS14}
L.H. Loomis and S.~Sternberg.
\newblock {\em Advanced Calculus}.
\newblock World Scientific, Singapore, 2014.

\bibitem{PT16}
N.~P\'{e}tr\'{e}lis and N.~Torri.
\newblock Collapse transition of the interacting prudent walk.
\newblock Preprint, (2016).

\bibitem{ST-phi4}
G.~Slade and A.~Tomberg.
\newblock Critical correlation functions for the $4$-dimensional weakly
  self-avoiding walk and $n$-component $|\varphi|^4$ model.
\newblock {\em Commun. Math. Phys.}, {\bf 342}:675--737, (2016).

\bibitem{Uelt02}
D.~Ueltschi.
\newblock A self-avoiding walk with attractive interactions.
\newblock {\em Probab.\ Theory\ Related\ Fields}, {\bf 124}:189--203, (2002).

\bibitem{Vand98}
C.~Vanderzande.
\newblock {\em Lattice Models of Polymers}.
\newblock Cambridge University Press, Cambridge, (1998).

\bibitem{WK74}
K.G. Wilson and J.~Kogut.
\newblock The renormalization group and the $\epsilon$ expansion.
\newblock {\em Phys. Rep.}, {\bf 12}:75--200, (1974).

\end{thebibliography}
